\documentclass[letterpaper,12pt]{article}

\usepackage{geometry,amssymb,amsmath,latexsym,amsthm}
\usepackage[T1]{fontenc}
\usepackage[utf8]{inputenc}
\usepackage{lmodern}
\newtoks\SupplementEveryDisplay
\SupplementEveryDisplay=\everydisplay
\usepackage{setspace}
\usepackage{indentfirst}
\usepackage{microtype}

\usepackage{enumitem}
\usepackage{booktabs}
\usepackage[font=small,labelfont=bf]{caption}
\usepackage[round,sort]{natbib}
\usepackage{xcolor}
\usepackage{hyperref}
\usepackage[nameinlink,noabbrev]{cleveref}

\hypersetup{
	colorlinks=true,
	linkcolor=blue,
	citecolor=blue,
	urlcolor=blue,
	pdfauthor={Paul H.Y. Cheung and Zichang Wang},
	pdftitle={Scrutiny and Conservatism}
}

\newtheoremstyle{blueplain}
{7pt}
{7pt}
{\itshape}
{}
{\bfseries\color{blue}}
{.}
{0.5em}
{}

\newtheoremstyle{bluedefinition}
{7pt}
{7pt}
{\normalfont}
{}
{\bfseries\color{blue}}
{.}
{0.5em}
{}

\theoremstyle{blueplain}
\newtheorem{theorem}{Theorem}
\newtheorem{proposition}{Proposition}
\newtheorem{lemma}{Lemma}

\theoremstyle{bluedefinition}
\newtheorem{definition}{Definition}
\newtheorem{axiom}{Axiom}

\newtheorem{remark}{Remark}

\theoremstyle{plain}
\newtheorem{suppaxiom}{Axiom}
\newtheorem{suppdefinition}{Definition}
\newtheorem{supplemma}{Lemma}
\newtheorem{suppauxproposition}{Auxiliary Proposition}
\newtheorem{suppproposition}{Proposition}
\newtheorem{suppremark}{Remark}

\crefname{theorem}{Theorem}{Theorems}
\Crefname{theorem}{Theorem}{Theorems}
\crefname{proposition}{Proposition}{Propositions}
\Crefname{proposition}{Proposition}{Propositions}
\crefname{lemma}{Lemma}{Lemmas}
\Crefname{lemma}{Lemma}{Lemmas}
\crefname{corollary}{Corollary}{Corollaries}
\Crefname{corollary}{Corollary}{Corollaries}
\crefname{claim}{Claim}{Claims}
\Crefname{claim}{Claim}{Claims}
\crefname{definition}{Definition}{Definitions}
\Crefname{definition}{Definition}{Definitions}
\crefname{axiom}{Axiom}{Axioms}
\Crefname{axiom}{Axiom}{Axioms}
\crefname{example}{Example}{Examples}
\Crefname{example}{Example}{Examples}
\crefname{remark}{Remark}{Remarks}
\Crefname{remark}{Remark}{Remarks}
\crefname{suppaxiom}{Axiom}{Axioms}
\Crefname{suppaxiom}{Axiom}{Axioms}
\crefname{suppdefinition}{Definition}{Definitions}
\Crefname{suppdefinition}{Definition}{Definitions}
\crefname{supplemma}{Lemma}{Lemmas}
\Crefname{supplemma}{Lemma}{Lemmas}
\crefname{suppauxproposition}{Auxiliary Proposition}{Auxiliary Propositions}
\Crefname{suppauxproposition}{Auxiliary Proposition}{Auxiliary Propositions}
\crefname{suppproposition}{Proposition}{Propositions}
\Crefname{suppproposition}{Proposition}{Propositions}
\crefname{suppremark}{Remark}{Remarks}
\Crefname{suppremark}{Remark}{Remarks}

\newcommand{\R}{\mathbb{R}}
\newcommand{\Acts}{\mathcal{F}}
\newcommand{\Menus}{\mathcal{M}}
\newcommand{\DeltaO}{\Delta(\Omega)}
\newcommand{\one}{\mathbf{1}}
\newcommand{\co}{\operatorname{co}}

\DeclareMathOperator{\supp}{supp}
\DeclareMathOperator*{\argmax}{argmax}

\setlist[enumerate]{label=(\alph*),leftmargin=2.2em}
\setlist[itemize]{leftmargin=1.6em}

\makeatletter
\newcommand{\BeginSupplementaryAppendix}{%
    \clearpage
    \newgeometry{margin=1in}%
    \setcounter{page}{1}%
    \setcounter{section}{0}%
    \setcounter{subsection}{0}%
    \setcounter{subsubsection}{0}%
    \setcounter{equation}{0}%
    \setcounter{suppaxiom}{0}%
    \setcounter{suppdefinition}{0}%
    \setcounter{supplemma}{0}%
    \setcounter{suppauxproposition}{0}%
    \setcounter{suppproposition}{0}%
    \setcounter{suppremark}{0}%
    \renewcommand\thesection{S\arabic{section}}%
    \renewcommand\theequation{S\arabic{equation}}%
    \renewcommand\thepage{S\arabic{page}}%
    \renewcommand\theHsection{supplement.\arabic{section}}%
    \renewcommand\theHsubsection{supplement.\arabic{section}.\arabic{subsection}}%
    \renewcommand\theHsubsubsection{supplement.\arabic{section}.\arabic{subsection}.\arabic{subsubsection}}%
    \renewcommand\theHequation{supplement.\arabic{equation}}%
    \let\axiom\suppaxiom
    \let\endaxiom\endsuppaxiom
    \let\definition\suppdefinition
    \let\enddefinition\endsuppdefinition
    \let\lemma\supplemma
    \let\endlemma\endsupplemma
    \let\auxproposition\suppauxproposition
    \let\endauxproposition\endsuppauxproposition
    \let\proposition\suppproposition
    \let\endproposition\endsuppproposition
    \let\remark\suppremark
    \let\endremark\endsuppremark
    \setlist[enumerate]{}%
    \setlist[itemize]{}%
    \everydisplay=\SupplementEveryDisplay
    \def\@afterindentfalse{\let\if@afterindent\iffalse}%
    \@afterindentfalse
    \def\NAT@sort{\z@}%
    \ignorespaces
}
\makeatother

\title{\normalfont Scrutiny and Conservatism}

\author{
	Paul H.Y. Cheung\thanks{
		Naveen Jindal School of Management, University of Texas at Dallas.
		Email: \href{mailto:paul.cheung@utdallas.edu}{paul.cheung@utdallas.edu}.
		Website: \url{https://www.paulcheung.net/}.
	}
	\and
	Zichang Wang\thanks{
		School of Economics and Wang Yanan Institute for Studies in Economics,
		Xiamen University.
		Email: \href{mailto:zichang.wang@xmu.edu.cn}{zichang.wang@xmu.edu.cn}.
		Website: \url{https://zichangwang.com/}.
	}
}

\date{Aug 28, 2026}

\begin{document}

\maketitle

\begin{abstract}
We study a setting in which an agent receives private information before choosing from a menu and anticipates hindsight scrutiny.  Such scrutiny creates a motive for conservatism toward menu expansions. Our key axiom, \textit{conservatism}, is a direct weakening of preference for flexibility: Adding an option is weakly beneficial whenever it leaves the menu's hindsight benchmark unchanged. Together with standard axioms, conservatism characterizes a scrutiny representation of preferences over menus in which the agent behaves as if she evaluates each menu by subtracting anticipated scrutiny from the material value of informed choice. Menu preference identifies a unique minimum pair of private information and scrutiny intensity, while incorporating subsequent stochastic-choice data pins down the actual pair. Applications to medical liability and corporate innovation illustrate the implications of our model for accountability design.

\end{abstract}



\newpage


\section{Introduction}
\label{sec:introduction}
Decisions made under uncertainty are often reevaluated once the uncertainty has been fully resolved. By then, an option reasonably rejected ex ante may appear obviously superior ex post, while the information that justified rejecting it may be difficult to verify. A forgone option can therefore become a source of scrutiny.

Consider a physician who combines documented test results with nonverifiable clinical judgment before choosing a treatment. After the patient's condition is fully understood, the physician may be asked: ``Why did you not choose the treatment that turned out to be best?'' Her truthful answer may be that the treatment was not optimal given the information available at the time. Yet because the complete clinical assessment that guided her decision is not fully verifiable, that answer may be difficult to substantiate. A potentially suitable treatment can therefore become an incriminating counterfactual. Similar tensions arise in many settings. A portfolio manager may choose among investment opportunities based on a proprietary report, while an employer may rely on private interviews to decide whom to hire.

In each of these settings, having more options creates both an opportunity and a liability. As in the standard preference-for-flexibility logic, a wider range of options allows the decision maker (DM) to choose differently depending on what she learns. Under hindsight scrutiny, however, an option that is not chosen does not simply disappear: once uncertainty is resolved, it may create a new basis for scrutiny. More options thus retain their familiar flexibility benefit while also exposing the DM to additional scrutiny. This trade-off leads to our central question: When are more options valuable under hindsight scrutiny?

In this paper, we study this question through preferences over menus of acts. An objective state of the world determines the consequences of each act. The DM first chooses a menu, then receives private information about the state and selects an act from the menu. Finally, the state is revealed. The DM anticipates being evaluated against the act in her menu that would have been best in the realized state, even though that act may have appeared to be suboptimal when she made her choice.

Our answer is a simple and intuitive behavioral postulate. To introduce it, we first extend the standard notion of statewise domination between acts to a domination relation between an act and a menu. We say that an act $g$ is \emph{dominated} by a menu $A$ if, in every state, some act in $A$ delivers an outcome at least as good as $g$. Note that adding a dominated act to a menu does not raise the menu's state-by-state hindsight benchmark. Yet the act can still be useful when private information is sufficiently imprecise, as it may serve as a ``safe'' option among more extreme acts in the menu. Our key axiom, \emph{conservatism}, therefore requires that adding a dominated act weakly improves the menu. Formally, it says $A\cup\{g\}\succsim A$ if $g$ is dominated by $A$. Thus, more options are unambiguously beneficial when they create no new basis for scrutiny. Conservatism is a direct weakening of preference for flexibility.

We now describe the functional form identified by our representation theorem. Let $u$ be an affine utility function over outcomes. If the DM chooses act $f$ from menu $A$ and state $\omega$ is subsequently revealed, her scrutiny cost is
\[
R(f,A,\omega)=K\left[\max_{g\in A}u(g(\omega))-u(f(\omega))\right],
\]
where $K\geq0$ measures the intensity of scrutiny. The bracketed term is the gap between the realized outcome of the chosen act and the best outcome that could have been obtained from the menu in that state. The DM's private information is a distribution $\mu$ over posterior beliefs. Before she selects an act, a posterior $p\in\Delta(\Omega)$ is drawn according to $\mu$. The DM evaluates a menu $A$ according to
\begin{equation}
\label{eq:intro-primitive}
V(A)=\int_{\Delta(\Omega)}\max_{f\in A}\sum_{\omega\in\Omega}p(\omega)\bigl[u(f(\omega))-R(f,A,\omega)\bigr]\,\mu(dp).
\end{equation}
Thus, the material benefit of informed choice is weighed against the scrutiny created by comparisons with what the menu could have delivered under perfect hindsight. We call this a \textit{scrutiny representation}.  

An important feature of \eqref{eq:intro-primitive} is that scrutiny affects the ex-ante value of a menu without distorting choice from it. For fixed $A$ and $p$, the hindsight benchmark is independent of which candidate act is being considered. The DM therefore still chooses an act maximizing posterior expected material utility. Her conservatism concerns which options she wants available, not how she chooses once the menu and information are fixed.

\medskip
\noindent\textbf{Example 1.}
Consider a physician who receives a diagnostic report before treating a patient. A clinical guideline contains an established treatment $f_1$ and a standard treatment $f_2$, giving the menu $B=\{f_1,f_2\}$. The physician could also make an experimental treatment $g$ available, giving $A=\{f_1,f_2,g\}$. Patients are equally likely to be type 1 or type 2, and the state-dependent utilities are

\begin{center}
\begin{tabular}{lcc}
\toprule
& Type 1 & Type 2 \\
\midrule
Established treatment $f_1$ & $10$ & $1$ \\
Standard treatment $f_2$ & $5$ & $6$ \\
Experimental treatment $g$ & $14$ & $-2.5$ \\
\bottomrule
\end{tabular}
\end{center}

The report is favorable or unfavorable with equal probability and yields a posterior probability of type 1 equal to $0.6$ or $0.4$, respectively. Under $B$, the physician chooses $f_1$ after a favorable report and $f_2$ after an unfavorable one. Under $A$, she chooses $g$ after a favorable report and $f_2$ after an unfavorable one. Direct calculation from \eqref{eq:intro-primitive} gives $V(A)=6.5-3.5K$ and $V(B)=6-2K$. 

The experimental treatment $g$ improves the expected material payoff from $6$ to $6.5$, but its availability also creates a more demanding state-by-state standard for evaluating the physician. Consequently, she prefers to stick to the clinical guideline and exclude the experimental treatment from consideration when $K>1/3$. The example captures the central trade-off: An option can improve informed choice while creating an even larger increase in scrutiny. \qed

In \Cref{sec:characterization}, we present our main characterization result. We show that conservatism, together with weak order, strong continuity, independence, and domination, is necessary and sufficient for a scrutiny representation. 

We then study the identification problem. Menu preference does not, in general, separately identify the agent's private information $\mu$ and scrutiny intensity $K$. Instead, it identifies a unique \textit{minimum} pair $(\underline\mu,\underline K)$: every other representation of the same (nonconstant) scrutiny preference must feature private information that is more informative, in the Blackwell sense, and a weakly larger scrutiny intensity. Moreover, we show that $\underline K>0$ if and only if the DM violates preference for flexibility. The minimum representation also allows us to obtain a comparative-statics result. In particular, we provide a behavioral characterization of comparative scrutiny aversion between two DMs who anticipate the same minimum private information but differ in their minimum scrutiny intensity $K$.

In \Cref{sec:stochastic-choice}, we turn to the DM's subsequent choices from menus. Because her choice depends on a privately observed signal, randomness in signal realizations generates stochastic choice data from the perspective of an outside observer, in the sense of \citet{Lu2016}. Although these data alone cannot reveal scrutiny aversion—the DM chooses optimally after every signal realization—they identify her private information. That information, together with the menu preference, selects a unique representation of the preference and hence pins down scrutiny intensity. The joint data therefore fully identify the DM's actual private information and scrutiny intensity.

We then turn from identification to characterization and ask when menu preference and subsequent stochastic-choice data can be generated by common primitives. We provide two cross-stage axioms. \emph{Singleton alignment} requires the two data sets to agree on the ranking of acts. \emph{Joint scrutiny consistency} requires a menu to be preferred over another menu whenever (i) stochastic choice reveals that it has a weakly greater informed-choice value and (ii) it has a weakly less demanding hindsight benchmark. Given the separate marginal representations, these two axioms are necessary and sufficient for a joint scrutiny representation.

In \Cref{sec:accountability-applications}, we apply the model to two delegated-decision settings with aligned material interests: a patient relies on a physician to consider and select a treatment, while a venture capitalist relies on a manager to generate and implement projects. Although the parties share the same material utility, the agent's anticipated scrutiny can distort these decisions. In the first application, blanket immunity weakens incentives for defensive diagnostic testing, consistent with \citet{FrakesGruber2019}, whereas a targeted safe harbor also increases the relative appeal of the protected treatment menu and may trade excessive diagnosis for excessive standardization. In the second application, intolerance of early failure discourages innovative projects, consistent with \citet{TianWang2014}, and creates a disclosure margin: a manager may explore broadly while keeping rejected experiments off the formal record. Thus, scrutiny can distort information acquisition, menu formation, and disclosure without any conventional conflict over material outcomes.

\subsection{Related Literature}
This paper first contributes to the literature on preferences over menus. Flexibility may be valuable because of uncertainty about future tastes (\citealp{Kreps1979} and \citealp{DekelLipmanRustichini2001}) or because the DM expects to learn about an objective state before choosing (\citealp{dillenberger2014theory}). Commitment (smaller menus) could be valuable if it protects a DM from temptation (\citealp{GulPesendorfer2001}) or regret (\citealp{Sarver2008}). As in \citet{dillenberger2014theory}, our DM learns about an objective state before choosing and remains materially optimal conditional on her information. Her conservatism arises because the statewise quality of forgone acts is observed after the state is revealed. While allowing preference for commitment, our key axiom directly identifies the type of flexibility that is unambiguously beneficial.

\citet{Sarver2008} axiomatizes anticipated regret from preferences over menus of lotteries. Sarver's Dominance axiom allows a weakly inferior alternative, according to the singleton ranking, to make a menu worse by raising ex-post regret. Such an alternative can never make the menu better.  Both papers use menu preferences to reveal costs created by forgone alternatives. Unlike Sarver's agent, however, ours receives private information before choosing, so choice need not follow the ex-ante singleton ranking; conservatism accordingly relies on statewise domination.

\citet{Kopylov2012} also allows an unchosen option to affect menu value through the menu's best normative option. The perfectionist set-betweenness axiom restricts comparisons among menus that share this option. There, menu dependence arises from temptation and self-evaluation. In our model, conditional choice is materially optimal and the relevant benchmark varies state by state.

The closest formal comparison concerning information is \citet{Pennesi2026}. The early-information representation there evaluates the improvement of informed choice relative to the act that is optimal at the prior. Our representation instead evaluates informed choice relative to the value attainable under perfect hindsight. Accordingly, Pennesi's interim rejoice-seeking axiom protects an expansion when the new act is no better than some existing act under the singleton ranking, whereas conservatism protects an expansion when the new act is dominated by the menu. The two axioms are not directly comparable and reflect different counterfactual benchmarks: Uninformed choice in his model and perfect hindsight in ours.

Our model also contributes to the literature on stochastic choice, particularly work that jointly studies menu preference and subsequent choice from menus. \citet{AhnSarver2013} study two behavioral primitives under uncertainty about future tastes: preferences over menus, in the sense of \citet{DekelLipmanRustichini2001}, and random choice from menus, in the sense of \citet{GulPesendorfer2006}. Their analysis maintains preference for flexibility, so larger menus are always weakly preferred. We depart from Ahn and Sarver along two dimensions. First, our DM faces objective uncertainty and receives private information about the state before choosing. Second, our model allows preference for commitment generated by hindsight scrutiny. As in Ahn and Sarver, combining menu preference with stochastic choice sharpens identification. In our setting, stochastic choice identifies the DM's private information; together with menu preference, this pins down scrutiny intensity and fully identifies the primitives of the model.

Our paper is also broadly related to the literature on hindsight bias. Once an outcome is known, people tend to overestimate how foreseeable it was ex ante \citep{Fischhoff1975,HawkinsHastie1990}. In a liability setting, \citet{KaminRachlinski1995} find that knowledge of an adverse outcome raises assessments of its prior likelihood and leads prior conduct to be judged more harshly. Our model does not require such a bias; hindsight bias would further strengthen the scrutiny mechanism we study.

The remainder of the paper is organized as follows. Section 2 introduces the model. Section 3 characterizes scrutiny preferences and studies uniqueness and comparative scrutiny aversion. Section 4 combines menu-preference data with choice-from-menu data to obtain identification and a joint axiomatic characterization. Section 5 develops the applications to medical liability and corporate innovation, and Section 6 concludes.


\section{Model}
\label{sec:model}
\noindent\textbf{Acts and menus.} Let \(\Omega\) be a finite set of states with $|\Omega|\geq 2$, and \(Z\) be a finite set of prizes with $|Z|\geq 2$. Let $X:=\Delta(Z)$ be the set of lotteries over prizes.\footnote{For any compact metric space $Y$, let $\Delta(Y)$ denote the set of probability measures on its Borel $\sigma$-algebra, endowed with the weak$^*$ topology. When $Y$ is finite, $\Delta(Y)$ is the set of all probability measures on $Y$.} An (Anscombe--Aumann) \textbf{act} is a function \(f:\Omega\to X\). Let \(\Acts\) denote the set of all acts, endowed with the Euclidean metric $d$. A menu is a nonempty compact subset of \(\Acts\). Let $\Menus$ denote the collection of all menus, endowed with the Hausdorff metric $d_h$.\footnote{This is defined by $d_h(A,B):=\max\{\max_{f\in A}\min_{g\in B}d(f,g),\max_{f\in B}\min_{g\in A}d(f,g)\}$.} Our primitive is a binary relation \(\succsim\) on \(\Menus\).

Mixtures are defined in the standard way. For \(\alpha\in[0,1]\), the act \(\alpha f+(1-\alpha)g\) delivers \(\alpha f(\omega)+(1-\alpha)g(\omega)\) in state \(\omega\). The mixture of
menus $A$ and $B$ with weight $\alpha\in[0,1]$ is defined by $\alpha A+(1-\alpha)B:=\{\alpha f+(1-\alpha)g:f\in A,\ g\in B\}$.

When no confusion can arise, we employ the common abuse of notation: $x$ also denotes the constant act that delivers lottery $x$ in every state, and $f$ also denotes the singleton menu containing only act $f$. With these conventions, we write $f(\omega)\succsim g(\omega)$ to mean that the singleton menu containing only the constant act $f(\omega)$ is preferred to the singleton menu containing only the constant act $g(\omega)$.

\medskip
\noindent\textbf{Belief and information.} We consider a decision maker (henceforth DM) who has a full-support prior belief $\pi$ in the interior of $\Delta(\Omega)$.\footnote{The full-support prior assumption is nonessential and purely for ease of exposition.} The DM's private information\footnote{Here and throughout, ``private'' means that the DM's complete contemporaneous information or posterior is not fully verifiable or reconstructible by an ex-post evaluator. Some components of the underlying evidence may nevertheless be observable or documented.} is modeled as a probability distribution over posteriors that is Bayes plausible with respect to $\pi$. That is, the DM's private information is some $\mu\in\Delta(\Delta(\Omega))$ satisfying $\pi=\int_{\Delta(\Omega)}p\mu(dp)$.

\subsection{Scrutiny}

The timeline of events is the following:
\begin{enumerate}
	\vspace{-0.1cm}\item[1.] First, the DM chooses a menu $A$, anticipating private information represented by $\mu$;
	\vspace{-0.15cm}\item[2.] Then, a posterior $p\in\Delta(\Omega)$ is drawn according to the information $\mu$;
	\vspace{-0.15cm}\item[3.] Knowing the posterior, the DM chooses an act $f\in A$;
	\vspace{-0.15cm}\item[4.] Finally, the state $\omega$ is revealed and the outcome $f(\omega)$ is realized.
\end{enumerate}

As explained in the Introduction, even though the DM might have chosen the best act from $A$ given her posterior belief $p$, her choice may be assessed differently after $\omega$ is revealed: ``Why did you not choose the act from $A$ that is best in state $\omega$?''

The exact form of scrutiny we identify is based on the difference between the assessments of the chosen act at belief $p$ and after state $\omega$ is revealed. Specifically, if we fix an affine utility function $u:X\to\R$ and a scalar $K\geq0$ measuring the intensity of scrutiny, the \textit{(ex-post) scrutiny cost} for having chosen act $f$ from menu $A$ after state $\omega$ is revealed is
$$R(f,A,\omega):=K\left[\max_{g\in A}u(g(\omega))-u(f(\omega))\right],$$
where the bracketed term will be referred to as the \textit{realized hindsight gap}.

Anticipating scrutiny, the DM evaluates a menu of acts by
\begin{equation}
\label{eq:scrutiny-primitive}
	V(A)=\int_{\Delta(\Omega)}\max_{f\in A}\left(\sum_{\omega\in\Omega}p(\omega)\big[(u(f(\omega))-R(f,A,\omega)\big]\right)\mu(dp).
\end{equation}
At each potential posterior $p\in\Delta(\Omega)$, the DM chooses an act with the highest expected net value, which is its expected material value minus its expected scrutiny cost.
\begin{definition}
\label{def:scrutiny}
	A binary relation $\succsim$ over $\Menus$ has a \textbf{scrutiny representation} if there exists a tuple $(u,\mu,K)$ such that $\succsim$ can be represented by $V$ as defined in equation~\eqref{eq:scrutiny-primitive} where $u:X\to\R$ is an affine utility function, $\mu\in\Delta(\Delta(\Omega))$ is a distribution over posteriors that induces a full-support prior, and $K\geq 0$ is a scalar representing the intensity of scrutiny.
\end{definition}

The representation accommodates two conceptually distinct sources of scrutiny. First, scrutiny may be internal, taking the form of ex-post regret or self-evaluation \citep{Bell1982,LoomesSugden1982}. After the state is revealed, the DM may experience a sense of loss from comparing her choice with the act that would have been best in that state, even when she remembers her information and recognizes that her choice was optimal given what she knew. Under this interpretation, privacy is not essential. Second, scrutiny may come from an outside evaluator. The evaluator may observe the realized state and the forgone alternatives without being able to reconstruct the DM's complete posterior because some of the evidence is not fully verifiable. Alternatively, the evaluator may observe much of the underlying evidence but evaluate the decision through a procedure that does not fully condition on it, placing weight instead on the realized outcome and the alternatives that could have been chosen. The scrutiny cost $R(f,A,\omega)$ captures the resulting hindsight comparison in reduced form. We remain agnostic about which channel is operative, and $K$ may reflect either source of scrutiny or a combination of the two.

Substituting the definition of $R$ into
equation \eqref{eq:scrutiny-primitive} yields the equivalent reduced form
\begin{equation}
\label{eq:scrutiny-reduced}
V(A)=(1+K)\int_{\Delta(\Omega)}\max_{f\in A}\left(\sum_{\omega\in\Omega}p(\omega)u(f(\omega))\right)\mu(dp)-K\sum_{\omega\in\Omega}\pi(\omega)\max_{f\in A}u(f(\omega)),
\end{equation}
where $\pi=\int_{\Delta(\Omega)}p\,\mu(dp)$. The first integral is the expected material utility achieved using the DM's private information. The second summation is referred to as the \textit{hindsight value} of a menu, that is, the value that the menu would deliver if the state were known before the act was selected.  When $K=0$, the representation reduces to the subjective-learning benchmark of \citet{dillenberger2014theory}. When $K>0$, a menu expansion can improve informed choice and nevertheless reduce menu value if it raises the hindsight value sufficiently more.

The reduced form also shows that the same act will be selected after a posterior is observed with or without scrutiny. For a fixed menu and posterior, the hindsight term is independent of the candidate act. Hence the DM chooses from $\argmax_{f\in A}\sum_{\omega\in\Omega}p(\omega)u(f(\omega)).$ The parameter \(K\) changes only the ex-ante ranking of menus, but not the choice from menus.


\section{Characterization}
\label{sec:characterization}

\subsection{Axioms}

We impose five axioms on preferences. The first three are standard in the setting of preferences over menus of lotteries, adapted to our framework with menus of acts.

\begin{axiom}[Weak Order]
	\label{ax:weak-order}
	The binary relation \(\succsim\) is complete and transitive.
\end{axiom}

\begin{axiom}[Strong Continuity]
	\label{ax:strongcontinuity}\
	\begin{enumerate}
		\item[1.] von Neumann--Morgenstern (vNM) Continuity: If $A\succ B\succ C$, then there exist $\alpha,\bar\alpha\in(0,1)$ such that 
        \vspace{-0.2cm}
        $$\alpha A+(1-\alpha)C\succ B\succ\bar\alpha A+(1-\bar\alpha)C.$$
		\item[2.] \vspace{-0.2cm} Lipschitz (L) Continuity: There exist menus $A^*,A_*\in\Menus$ and $\gamma>0$ such that for every $A,B\in\Menus$ and $\alpha\in(0,1)$ with $d_h(A,B)\leq\alpha/\gamma$,
        \vspace{-0.2cm}
		$$(1-\alpha)A+\alpha A^*\succsim(1-\alpha)B+\alpha A_*.$$
	\end{enumerate}
\end{axiom}

\begin{axiom}[Independence]
	\label{ax:independence}
	If $A\succ B$, then for all $C\in\Menus$ and all $\alpha\in(0,1]$,
	$$\alpha A+(1-\alpha)C\succ\alpha B+(1-\alpha)C.$$
\end{axiom}
Because the representation permits a potentially infinite subjective state space and a signed measure, standard vNM-continuity has to be strengthened to strong continuity. For a more detailed discussion on L-Continuity, see \citet[henceforth DLR]{DekelLipmanRustichini2001} and a relevant corrigendum \citet[henceforth DLRS]{DekelLipmanRustichiniSarver2007}.

To proceed, we introduce the notion of ``statewise domination.''
\begin{definition}[Statewise Domination]
\label{def:statewisedomination}We say an act $g$ is \textit{statewise dominated} by a menu $A$ if for any $\omega\in\Omega$, there exists $f\in A$ such that $f(\omega)\succsim g(\omega)$. We say that a menu $B$ is \textit{statewise dominated} by another menu $A$ if every act in $B$ is statewise dominated by $A$.
\end{definition}

Under our definition, an act $g$ could be statewise dominated by a menu $A$ even when $g$ is not statewise dominated by \textit{any single act} in $A$. This definition only reduces to the standard state-by-state domination between acts if both $A$ and $B$ are singletons. That is, $g$ is statewise dominated by $f$ if and only if $f(\omega)\succsim g(\omega)$ for every $\omega$. We henceforth suppress the qualifier and just write ``$g$ is dominated by $A$'' to indicate statewise domination.

Our next axiom deals with acts satisfying state-by-state domination.

\begin{axiom}[Domination]
	\label{ax:domination}
	For any $f,g\in\Acts$ such that $g$ is dominated by $f$:
	\begin{enumerate}
		\item $f\succsim g$;
		\item If, in addition, $f\in A$ for some $A\in\Menus$, then $A\sim A\cup\{g\}$;
		\item If, in addition, $f(\omega^*)\succ g(\omega^*)$ for some $\omega^*\in\Omega$, then $f\succ g$.
	\end{enumerate}
\end{axiom}
This axiom combines three common requirements regarding preference over menus of acts and preference over acts themselves. Clause (a) is the standard monotonicity requirement for a preference over acts. Clause (b) is the ``domination'' axiom used in \citet[henceforth DLST]{dillenberger2014theory} and the source of the axiom's name. If $\succsim$ satisfies preference for flexibility (which was assumed in DLST), then (b) implies (a). Since our motivation is to weaken preference for flexibility, (a) is not automatically implied by (b). Clause (c) is the standard Anscombe-Aumann requirement to make sure no state is a null state. As we explained, clause (c) is nonessential.

Our key axiom identifies the exact type of flexibility that is valuable to a DM.
\begin{axiom}[Conservatism]
	\label{ax:PDF}
	If an act $g$ is dominated by a menu $A$, then $A\cup\{g\}\succsim A$.
\end{axiom}

\textit{Conservatism} is a direct weakening of preference for flexibility. Adding a dominated act gives the DM an additional option that can be chosen without creating a harsher benchmark against which the DM is scrutinized. Note that, in combination with \Cref{ax:domination}(b), adding $g$ to the menu $A$ is strictly welfare-improving only if $g$ is not dominated by any particular $f\in A$. In such a case, under the model, there must exist nondegenerate posterior beliefs under which $g$ is chosen, with $g$ serving as a ``safe'' or conservative choice relative to the other acts in $A$.

Moreover, the axiom does not rule out the possibility of improving a menu by adding an undominated act. Adding an undominated act could still be beneficial if its contribution to informed choice is large enough to compensate for the more demanding benchmark it creates.

\subsection{Representation Theorem}
The following is our main representation theorem:
\begin{theorem}
	\label{thm:characterization}
	A preference $\succsim$ has a scrutiny representation if and only if it satisfies weak order, strong continuity, independence, domination, and conservatism.
\end{theorem}
Given this result, we refer to a preference that satisfies weak order, strong continuity, independence, domination, and conservatism as a \textit{scrutiny preference}. Here we provide an outline of the proof of the theorem. The formal proof is in Section~\ref{app:proof-characterization} of the Appendix.

The first step in establishing Theorem 1 is to note that the scrutiny representation is a special case of a more general class of additive expected-utility (EU) representations, which is the counterpart to the additive EU representation characterized in DLR and DLRS for preferences over menus of lotteries.

\begin{definition}
	An \textbf{additive EU representation} is a pair $(u,\eta)$ that consists of a 	finite (and countably additive) signed Borel measure $\eta$ on $\Delta(\Omega)$ and an affine utility function $u:X\to\R$ such that $\succsim$ is represented by the function $W:\Menus\to\R$ defined by
	\begin{equation}
		W(A)=\int_{\Delta(\Omega)}\max_{f\in A}\left(\sum_{\omega\in\Omega}p(\omega)u(f(\omega))\right)\eta(dp),
	\end{equation}
	where  $\eta(\Delta(\Omega))=1$ and its barycenter
	 $\int_{\Delta(\Omega)}p\eta(dp)$ is in the interior of $\Delta(\Omega)$.
\end{definition}

The following is a representation theorem for the additive EU representation.

\begin{proposition}
	\label{prop:additive-foundation}
	A preference $\succsim$ has an additive EU representation if and only if it satisfies
	weak order, strong continuity, independence, and domination.
\end{proposition}
The proof of Proposition \ref{prop:additive-foundation} is relegated to the supplementary appendix. It closely follows the DLR--DLRS characterization and relies heavily on its intermediate results. \citet[Proposition 1]{Pennesi2026} provides a closely related result. Our proof instead uses the geometric translation technique introduced in DLST and clarifies the role of each clause of domination.

In light of Proposition 1, the necessity of Axioms 1 to 4 follows from verifying that every scrutiny representation is an additive EU representation, and we explicitly verify that a scrutiny representation induces a preference satisfying conservatism.

To show sufficiency, we first obtain an additive EU representation $(u,\eta)$ through Axioms 1 to 4. We then show that conservatism dictates that the negative weights of $\eta$ must be concentrated on the set of degenerate posteriors on $\Omega$, which allows us to construct $(\mu,K)$, where $\mu$ is a distribution over posteriors obtained as a combination of $\eta$ and the distribution over degenerate posteriors, with the weight depending on $K$.

\subsection{Uniqueness of the Representation}
\label{subsec:menu-identification}

We now study the extent to which the primitives of a scrutiny representation are identified from menu preference. We begin by establishing uniqueness of the underlying additive EU representation. To do so, we first rule out the degenerate case in which the DM is indifferent between all menus. We say that a preference $\succsim$ on $\Menus$ is \textit{constant} if $A\sim B$ for any $A,B\in\Menus$. A preference on $\Menus$ is \textit{nonconstant} if it is not constant. We reserve the terms ``trivial'' and ``nontrivial'' for a more restricted meaning.

\begin{axiom}[Lottery Nontriviality]
\label{axiom:lotterynontriviality}
There exists $x,y\in X$ with $\{x\}\succ\{y\}$.
\end{axiom}

Lottery nontriviality is closely related to nonconstancy of $\succsim$.
\begin{lemma}
\label{lem:lottery-nontriviality-equivalence}
Suppose $\succsim$ has an additive EU representation. Then $\succsim$ is nonconstant if and only if it satisfies lottery nontriviality.
\end{lemma}
See \Cref{app:proof-lottery-nontriviality-equivalence} for the proof.

For two real-valued functions $u,v:Y\to\R$, we write $u\approx v$ to denote that there exists $a>0$ and $b\in\R$ such that $v=au+b$. 

\begin{proposition}
\label{prop:additive-uniqueness}
Suppose $\succsim$ satisfies Lottery Nontriviality. Two additive EU representations  $(u,\eta)$ and $(u',\eta')$  both represent $\succsim$ if and only if $u\approx u'$ and $\eta=\eta'$.
\end{proposition}
The proof is relegated to the supplementary appendix. Thus, for a nonconstant scrutiny preference, menu preference uniquely identifies the utility index, up to the usual positive affine transformation, and the signed measure $\eta$. This, however, does not yet imply uniqueness of the scrutiny representation. The same signed measure may correspond to different combinations of private information $\mu$ and scrutiny intensity $K$.

To understand this remaining indeterminacy, we first ask a simpler question: When can the observed menu preference be represented without scrutiny at all? The answer is closely related to preference for flexibility.

\begin{axiom}[Preference for Flexibility]
\label{axiom:flexibility}
If $A\supseteq B$, then $A\succsim B$.
\end{axiom}

In the following lemma, we establish that a DM with a scrutiny preference that also satisfies preference for flexibility behaves as if she anticipates zero scrutiny.

\begin{lemma}
\label{lem:zero-scrutiny}
A scrutiny preference $\succsim$ satisfies preference for flexibility if and only if $\succsim$ has a scrutiny representation $(u,\mu,0)$.
\end{lemma}
See \Cref{app:proof-zero-scrutiny} for a proof.

In light of Lemma \ref{lem:zero-scrutiny}, we say that a scrutiny preference $\succsim$ is \textit{trivial} if it satisfies preference for flexibility. We say that a scrutiny preference $\succsim$ is \textit{nontrivial} if it is not trivial. Since violation of preference for flexibility requires the existence of some $B\subseteq A$ satisfying $B\succ A$, a nontrivial scrutiny preference must also be nonconstant.

The lemma identifies exactly when zero scrutiny is consistent with the observed menu preference. It does not, however, imply that scrutiny intensity is uniquely identified as zero. A preference satisfying preference for flexibility may also admit a scrutiny representation $(u',\mu',K')$ with $K'>0$. Thus, even after the additive EU representation is uniquely identified, its decomposition into private information and scrutiny intensity need not be unique. 

We now characterize the full extent of this nonuniqueness. Although menu preference need not select a unique pair $(\mu,K)$, the admissible pairs turn out to have a particularly simple structure. For a scrutiny preference $\succsim$, define $\mathcal P(\succsim)$ by
$$\mathcal P(\succsim):=\{(\mu,K)\mid \text{there exists $u$ such that }(u,\mu,K)\text{ can represent $\succsim$}\}.$$
That is, $\mathcal P(\succsim)$ collects all information-intensity pairs that can be part of a scrutiny representation for $\succsim$. For two Bayes-plausible posterior distributions with a common prior, we say that $\mu$ is \emph{Blackwell more informative} than $\lambda$ if $\int\varphi(p)\,\mu(dp)\geq\int\varphi(p)\,\lambda(dp)$ for every continuous convex function $\varphi:\Delta(\Omega)\to\R$. 

Our next lemma establishes a minimum of $\mathcal P(\succsim)$ for lottery-nontrivial $\succsim$.

\begin{lemma}
\label{lem:minimum-pair}
Let $\succsim$ be a nonconstant scrutiny preference. Then there exists $\underline\mu(\succsim)\in\Delta(\Delta(\Omega))$ and $\underline K(\succsim)\geq 0$ such that:
\begin{enumerate}
	\item $(\underline\mu(\succsim),\underline K(\succsim))\in\mathcal P(\succsim)$; and
	\item For any $(\mu,K)\in\mathcal P(\succsim)$, $\mu$ is Blackwell more informative than $\underline\mu(\succsim)$ and $K\geq\underline K(\succsim)$.
\end{enumerate} 
Moreover, $\underline K(\succsim)>0$ if and only if $\succsim$ is nontrivial.
\end{lemma}
See \Cref{app:proof-minimum-pair} for a proof.

The lemma shows that, among all scrutiny representations of a nonconstant preference, there is a least informative information structure and a smallest scrutiny intensity. Any other representation must involve private information that is weakly more informative than $\underline\mu(\succsim)$ and scrutiny intensity weakly greater than $\underline K(\succsim)$. Intuitively, more informative private information raises the value of informed choice toward the perfect-hindsight benchmark, making greater scrutiny compatible with the same observed menu preference.

The minimum pair is unique. To see this, suppose both $(\mu,K)$ and $(\mu',K')$ satisfy the two properties in \Cref{lem:minimum-pair}. Applying the second property in both directions implies that $\mu$ and $\mu'$ are Blackwell equivalent\footnote{Because information is represented as distributions over posteriors with a common prior, Blackwell dominance coincides with convex order. Hence mutual Blackwell dominance implies $\mu=\mu'$.} and that $K=K'$. Hence, each nonconstant scrutiny preference $\succsim$ has a unique minimum pair, denoted by
$(\underline\mu(\succsim),\underline K(\succsim))$. We can therefore associate every nonconstant scrutiny preference with a canonical scrutiny representation.

\begin{definition}
A \textbf{minimum representation} for a nonconstant scrutiny preference $\succsim$ is a tuple $(u,\mu,K)$ such that $(u,\mu,K)$ represents $\succsim$ and $(\mu,K)=(\underline\mu(\succsim),\underline K(\succsim))$.
\end{definition}

The minimum representation is unique up to the positive affine normalization of utility.

\begin{theorem}
	\label{thm:menu-uniqueness}
	Suppose $(u,\mu,K)$ and $(u',\mu',K')$ are minimum representations. They represent the same nonconstant scrutiny preference if and only if $u\approx u'$, $\mu=\mu'$, and $K=K'$.
\end{theorem}
See \Cref{app:proof-menu-uniqueness} for a proof.

\subsection{Comparative Scrutiny Aversion}
\label{subsec:comparative-statics}

A natural way to compare scrutiny aversion across DMs is to compare their willingness to commit. Greater scrutiny makes menus less attractive relative to singleton acts, so a more scrutiny-averse DM should be more willing to give up flexibility. Commitment behavior alone, however, does not isolate scrutiny intensity. A DM may also value a menu less because her private information is less useful for choosing among its acts. A behavioral comparison of scrutiny aversion must therefore both order willingness to commit and control for differences in private information.

We focus on two nontrivial scrutiny preferences $\succsim_1$ and $\succsim_2$. Recall the notion of statewise domination between menus in \Cref{def:statewisedomination}: A menu $A$ is dominated by $B$ under both $\succsim_1$ and $\succsim_2$ if for every $f\in A$ and every $\omega\in\Omega$, there exist $g_1,g_2\in B$ such that $g_1(\omega)\succsim_1f(\omega)$ and $g_2(\omega)\succsim_2f(\omega)$.

\begin{definition}[More Scrutiny Aversion]
\label{def:more-scrutiny}
Let $\succsim_1$, $\succsim_2$ be two nontrivial scrutiny preferences. We say that \(\succsim_1\)  \emph{is more scrutiny averse} than
\(\succsim_2\) if the following conditions hold:
\begin{enumerate}
\item[(a).] For every \(A\in\Menus\) and
every \(f\in\Acts\),
$
\{f\}\succsim_2 A\Longrightarrow
\{f\}\succsim_1 A.
$
\item[(b).] If $A$ and $B$ dominate each other under both $\succsim_1$ and $\succsim_2$, then $A\succsim_1B\iff A\succsim_2B$.
\end{enumerate}
\end{definition}

Condition (a) orders the two DMs by their willingness to commit: whenever DM 2 weakly prefers committing to an act rather than retaining a menu, DM 1 does as well. The singleton-versus-menu comparison used here is a commonly used comparative device in the menu-choice literature.\footnote{See, for example, \citet{Sarver2008}, \citet{DekelLipman2012}, \citet{dillenberger2014theory}, and \citet{deOliveiraDentiMihmOzbek2017}. \citet{DekelLipman2012} discuss the relation of this comparative criterion to \citet{Ahn2008} and \citet{HigashiHyogoTakeoka2009}.} In our setting, it captures the part of comparative scrutiny aversion revealed through willingness to give up flexibility.

Condition (b) serves a different purpose and is unique to our setting. It compares the two DMs only across pairs of menus with the same statewise maximal utility. Once their singleton preferences have been aligned, mutual domination of $A$ and $B$ implies that the two menus generate the same perfect-hindsight benchmark. Their ranking then depends only on the value generated by private information, up to a positive multiplicative factor involving scrutiny intensity. Requiring the two DMs to agree on all such comparisons controls for differences in private information.

\begin{theorem}[Comparative Scrutiny Aversion]
\label{thm:comparative-scrutiny}
Suppose that \(\succsim_i\)  has a minimum representation $(u_i, \mu_i, K_i)$ and $K_i>0$ for $i=1,2$. The following statements are equivalent:
\begin{enumerate}
\item[(1).] \(\succsim_1\) is more scrutiny averse than \(\succsim_2\).
\item[(2).] $u_1\approx u_2$, $\mu_1=\mu_2$, and $K_1\geq K_2$.
\end{enumerate}
\end{theorem}
See \Cref{app:proof-comparative-scrutiny} for a proof.

The theorem gives a behavioral interpretation of the scrutiny intensity in the minimum representation. Once tastes and private information are held fixed, greater scrutiny aversion is equivalent to a larger $K$. Thus, menu preference provides a behavioral ordering of the minimum scrutiny intensity across DMs.

\section{Stochastic Choice and Identification}
\label{sec:stochastic-choice}

\Cref{sec:characterization} shows that menu preference does not, in general, separately identify the DM's private information and scrutiny intensity: different combinations of \(\mu\) and \(K\) can generate the same menu preference. We now add the DM's subsequent stochastic choice from menus. Because scrutiny does not affect posterior-optimal choice, stochastic choice isolates the information used at the choice stage. Once this information is identified, menu preference pins down the remaining scrutiny intensity.

We proceed in three steps. We first define joint scrutiny representations. We then show that stochastic choice identifies the utility index and posterior distribution and directly recovers the informed-choice value of each menu. Combining these objects with menu preference identifies scrutiny intensity. Finally, we characterize when separately representable menu preference and stochastic choice admit a joint representation. Two cross-stage behavioral conditions provide the characterization.

\subsection{Stochastic Choice and Joint Scrutiny Representations}
\label{subsec:joint-representation}

Let $\Menus^f\subseteq\Menus$ be the collection of finite menus. A \emph{random choice rule} is a family $\rho=(\rho_A)_{A\in\Menus^f}$ such that $\rho_A\in\Delta(A)$ for every $A\in\Menus^f$. Thus, $\rho_A(f)$ is the probability that act $f$ is selected from $A$. For $B\subseteq A$, write $\rho_A(B):=\sum_{f\in B}\rho_A(f).$

Fix a nonconstant affine utility function $u:X\to\R$. Conditional on posterior $p\in\Delta(\Omega)$, the expected utility of $f$ is $U(f,p):=\sum_{\omega\in\Omega}p(\omega)u(f(\omega))$, and the posterior-optimal acts in a finite menu $A$ form the nonempty set
\[
	C_A(p;u):=\argmax_{f\in A}U(f,p).
\]
Posterior indifferences are resolved by a tie-breaking kernel. Formally, for each finite menu $A$, a map $\tau_A:\Delta(\Omega)\to\Delta(A)$ is an \emph{admissible tie-breaking kernel for menu $A$} if every $p\mapsto\tau_A(f\mid p)$ is Borel measurable and
\[
	\tau_A(f\mid p)>0
	\quad\Longrightarrow\quad
	f\in C_A(p;u).
\]
The kernel may select a tied maximizer deterministically or randomize among tied maximizers. It may also vary with the menu. A \emph{tie-breaking rule} is a family $\tau:=(\tau_A)_{A\in\Menus^f}$ such that $\tau_A$ is an admissible tie-breaking kernel for each $A\in\Menus^f$.

\begin{definition}[Information Representation]
	\label{def:information-representation}
	A random choice rule $\rho$ has an \emph{information representation} if there exist a nonconstant affine $u:X\to\R$, a distribution over posteriors $\mu\in\Delta(\Delta(\Omega))$, and a tie-breaking rule $\tau$ such that
	\begin{equation}
		\rho_A(f)
		=
		\int_{\Delta(\Omega)}\tau_A(f\mid p)\,\mu(dp)
		\qquad
		\text{for every }A\in\Menus^f
		\text{ and every }f\in A.
		\label{eq:information-representation}
	\end{equation}
	We refer to $(u,\mu,\tau)$ as the primitives of the information representation and suppress the tie-breaking rule when no confusion can arise.
\end{definition}

This formulation is closely related to \citet{Lu2016}, but models posterior ties explicitly through menu-dependent kernels. This allows us to dispense with Lu's common-measurability convention and regularity restriction. In this respect, our treatment is closer to \citet{AhnSarver2013}, who take random choice from menus of lotteries as primitive.

We can now introduce the joint scrutiny representation.

\begin{definition}[Joint Scrutiny Representation]
	\label{def:joint-scrutiny-representation}
	A \emph{joint scrutiny representation} for a pair $(\succsim,\rho)$ is a tuple $(u,\mu,K,\tau)$ such that $(u,\mu,K)$ is a scrutiny representation for $\succsim$ and $(u,\mu,\tau)$ is an information representation for $\rho$.
\end{definition}

The definition requires the information anticipated when the menu is evaluated to be the information used when an act is selected. The scrutiny parameter does not enter posterior maximization. Hence, holding $(u,\mu,\tau)$ fixed, every $K\geq0$ generates exactly the same stochastic choice rule.

\subsection{Identification from Stochastic Choice}
\label{subsec:choice-identification}

We next establish the identification result needed for the uniqueness of the joint representation. The argument adapts the test-act insight of \citet{Lu2016}, while allowing the explicit tie-breaking kernels in \Cref{def:information-representation}.

\begin{proposition}
	\label{prop:binary-choice-identification}
	If two information representations $(u,\mu,\tau)$ and $(u',\mu',\tau')$ induce the same choice probabilities on every binary menu, then $u'\approx u$ and $\mu'=\mu$.
\end{proposition}
See \Cref{app:proof-binary-choice-identification} for a proof.

The proof provides a procedure for identification as in \citet{Lu2016}. Choices between constant acts first identify $u$ up to a positive affine transformation. Normalize $u$ so that $u(X)=[0,1]$, choose $\bar x,\underline x\in X$ satisfying $u(\bar x)=1$ and $u(\underline x)=0$, and, identifying them with constant acts, define the \emph{test act} $x_t:=t\bar x+(1-t)\underline x$ for every $t\in[0,1]$. For each $v\in[0,1]^\Omega$, define $f_v(\omega):=x_{v_\omega}$, so that $U(f_v,p)=p\cdot v$. Choice from the binary menu $\{f_v,x_t\}$ reveals the distribution of the posterior projection $p\cdot v$: regardless of tie-breaking,
\[
	\mu\bigl(\{p:p\cdot v>t\}\bigr)
	\leq
	\rho_{\{f_v,x_t\}}(f_v)
	\leq
	\mu\bigl(\{p:p\cdot v\geq t\}\bigr).
\]
Thus, binary choice probabilities identify the distribution of $p\cdot v$ at all continuity points. Varying $v$ and applying the Cramér--Wold theorem then identifies $\mu$. However, note that even equality of the entire induced random choice rules need not identify $\tau$, not even $\mu$-almost surely. Nevertheless, identification of $u$ and $\mu$ is sufficient for our purposes. 

We can now present the identification result for the joint scrutiny representation. We call a pair $(\succsim,\rho)$ \emph{nontrivial} if its scrutiny preference $\succsim$ is nontrivial.

\begin{theorem}
	\label{thm:joint-uniqueness}
	If two joint scrutiny representations $(u,\mu,K,\tau)$ and $(u',\mu',K',\tau')$ represent the same nontrivial pair $(\succsim,\rho)$, then $u'\approx u$, $\mu'=\mu$, and $K'=K$.
\end{theorem}

By \Cref{prop:binary-choice-identification}, any two joint scrutiny representations share the same utility index, up to a positive affine transformation, and the same distribution of posterior beliefs. Once these objects are fixed, a nontrivial menu preference also pins down $K$. Thus, unlike menu preference alone, which identifies only the minimum intensity $\underline K$, the joint data identify the actual scrutiny intensity $K$. The tie-breaking rule $\tau$ remains unidentified. The complete proof is in \Cref{app:proof-joint-uniqueness}.

\subsection{Characterization of the Joint Scrutiny Representation}

This subsection characterizes when menu preference and stochastic choice are generated by the same utility index and private information. Throughout, suppose that the menu preference $\succsim$ is a nonconstant scrutiny preference and that the random choice rule $\rho$ has an information representation.

Note that the characterization requires us to compare the value of menu based on the information representation. Following \citet{Lu2016}, the same test-act construction from the previous subsection also recovers the informed-choice value of every menu. Continue to use the normalization $u(X)=[0,1]$ and the test acts $x_t$ defined before. For any menu $A$ and posterior $p$, define the \emph{posterior value} of $A$ by
\[
	\sigma_A^u(p)
	:=
	\max_{f\in A}U(f,p)
	=
	\max_{f\in A}
	\sum_{\omega\in\Omega}p(\omega)u(f(\omega)).
\]
Thus, $\sigma_A^u(p)$ is the utility obtained by choosing optimally from $A$ under posterior $p$.

For a finite menu $A$, define
\begin{equation}
	T_A^\rho(t)
	:=
	\rho_{A\cup\{x_t\}}\bigl(A\setminus\{x_t\}\bigr),
	\qquad t\in[0,1].
	\label{eq:test-curve}
\end{equation}
The set difference matters only when $x_t\in A$, which can occur for at most finitely many $t$ and therefore does not affect the integral below. For all other $t$, $T_A^\rho(t)$ is the probability that the DM chooses from $A$ rather than the constant act of utility $t$.

Conditional on posterior $p$, the DM necessarily chooses from $A$ if $\sigma_A^u(p)>t$ and necessarily chooses $x_t$ if $\sigma_A^u(p)<t$. Tie-breaking can matter only when $\sigma_A^u(p)=t$. Consequently,
\[
	\mu\bigl(\{p:\sigma_A^u(p)>t\}\bigr)
	\leq
	T_A^\rho(t)
	\leq
	\mu\bigl(\{p:\sigma_A^u(p)\geq t\}\bigr).
\]
These bounds imply that $T_A^\rho$ is weakly decreasing and hence Borel measurable.

The following result adapts the test-function integration argument of \citet[Theorem~2]{Lu2016} to our explicit treatment of posterior indifferences.

\begin{lemma}
	\label{lem:recover-choice-value}
	If $(u,\mu,\tau)$ represents $\rho$ and $u$ is normalized to $u(X)=[0,1]$, then, for every finite menu $A$,
	\[
		\int_0^1T_A^\rho(t)\,dt
		=
		\int_{\Delta(\Omega)}\sigma_A^u(p)\,\mu(dp),
	\]
	independently of the tie-breaking kernels.
\end{lemma}
See \Cref{app:proof-choice-value} for a proof.

Accordingly, define the informed-choice value recovered from $\rho$ by
\begin{equation}
	I_\rho(A):=\int_0^1T_A^\rho(t)\,dt
	\label{eq:choice-value}
\end{equation}
for every finite menu $A$. The lemma shows that $I_\rho(A)$ is identified from $\rho$ and equals the menu's expected value under posterior-optimal choice.

Note that the Information representation is defined on finite menus, $\Menus^f$, while the scrutiny presentation is defined on compact menus, i.e. $\Menus$. Nonetheless, since finite menus are dense in $\Menus$ and the mapping $A\mapsto\int_{\Delta(\Omega)}\sigma_A^u(p)\,\mu(dp)$ is continuous in the Hausdorff metric, $I_\rho$ has a unique continuous extension from $\Menus^f$ to $\Menus$. We use the same notation $I_\rho$ for this extension and use this extension for our characterization. Note that the informed-choice value induces an auxiliary preference over menus: Write $A\succsim_\rho B$ if $I_\rho(A)\geq I_\rho(B)$. 

\begin{axiom}[Singleton Alignment (SA)]
	\label{ax:singleton-alignment}
	For any $f,g\in\Acts$, $\{f\}\succsim\{g\}\iff\{f\}\succsim_\rho\{g\}.$
\end{axiom}

Singleton menus generate no scrutiny gap. SA therefore requires the material ranking of acts revealed by menu preference to coincide with the ranking recovered from stochastic choice. It imposes no corresponding agreement over nonsingleton menus.

\begin{axiom}[Joint Scrutiny Consistency (JSC)]
	\label{ax:joint-scrutiny-consistency}
	For all $A,B\in\Menus$, if $A\succsim_\rho B$ and $A$ is dominated by $B$, then $A\succsim B$.
\end{axiom}

The relation $\succsim_\rho$ ranks menus by informed-choice value. Thus, $A\succsim_\rho B$ means that $A$ offers weakly higher material value. If $A$ is also dominated by $B$, then it generates a weakly less demanding hindsight benchmark. When both comparisons favor $A$, the scrutiny preference must rank $A$ weakly above $B$. The following theorem shows that JSC is both necessary and sufficient for a joint representation.

\begin{theorem}
	\label{thm:joint-characterization}
	Suppose that $\succsim$ is a nonconstant scrutiny preference, $\rho$ has an information representation, and the pair $(\succsim,\rho)$ satisfies SA. Then $(\succsim,\rho)$ has a joint scrutiny representation if and only if it satisfies JSC.
\end{theorem}
See \Cref{app:proof-joint-characterization} for a proof.

Necessity follows because common primitives align singleton values and make the informed-choice ranking recovered from \(\rho\) compatible with the scrutiny preference. Conversely, SA aligns the utility indices underlying the two data sets, while JSC ensures that the informed-choice functional recovered from \(\rho\) can serve as the information component of a scrutiny representation of \(\succsim\). The formal argument is provided in \Cref{app:proof-joint-characterization}.

\section{Applications}
\label{sec:accountability-applications}

In our applications, we demonstrate how scrutiny can distort behavior in a principal-agent setting even when the two parties' material interests are perfectly aligned. A patient relies on an informed physician to consider a range of potential treatments and select one; a venture capitalist (VC) relies on an informed manager to generate and evaluate potential innovative projects before implementing one. Even when the parties share the same material utility, distortions can arise because the agent also anticipates hindsight scrutiny---whether externally imposed or internally perceived. We use this mechanism to provide a complementary explanation for the evidence on defensive diagnostic testing in \citet{FrakesGruber2019} and the association between failure tolerance and innovation in \citet{TianWang2014}, and then use the corresponding extensions of the model to analyze guideline-based malpractice safe harbors and mandatory disclosure of experimentation.


\subsection{From Blanket Immunity to Targeted Safe Harbors}
\label{subsec:medicine-safe-harbors}

\citet{FrakesGruber2019} exploit variation in malpractice liability within the U.S. Military Health System. Active-duty patients treated at military facilities cannot sue for negligent care, whereas military dependents and patients treated at civilian facilities retain malpractice protections. They find that immunity reduces diagnostic testing by roughly 22 percent, without measurable deterioration in mortality, unplanned readmission, or patient-safety measures. We interpret immunity as a reduction from $K_1$ to $K_0<K_1$ that applies across treatment menus. Because professional, institutional, and reputational review may remain, $K_0$ need not be zero. This interpretation does not claim that hindsight scrutiny is the unique explanation for the evidence; it asks whether the observed diagnostic response is consistent with our mechanism.

The patient delegates diagnostic acquisition, consideration of treatment alternatives, and subsequent treatment choice to the physician. We abstract from conflicting treatment preferences by assigning the physician and patient a common material utility $u$. The relevant margin is therefore information acquisition rather than overtreatment: scrutiny can affect which diagnostic technology is selected while leaving treatment posterior optimal conditional on the resulting information.

Let $M$ be a finite menu of treatments and $\pi$ the prior over patient subtypes. A diagnostic technology is a Bayes-plausible distribution $\mu\in\Delta(\DeltaO)$ over posteriors with barycenter $\pi$. The physician chooses $\mu$ from a finite feasible set $\mathcal E$. Define
\begin{align}
	I_\mu(M)
	&:=
	\int_{\DeltaO}
	\max_{f\in M}
	\sum_{\omega\in\Omega}p(\omega)u(f(\omega))\,\mu(dp),
	\ \ \text{ and }\ \ H(M)
	:=
	\sum_{\omega\in\Omega}\pi(\omega)
	\max_{f\in M}u(f(\omega)).
	\label{eq:med-values}
\end{align}
If $c(\mu)$ is the resource cost of diagnostic technology $\mu$, the physician's optimized value is
\begin{equation}
	\mathcal V_K(M)
	:=
	\max_{\mu\in\mathcal E}
	\left\{
	(1+K)I_\mu(M)-K H(M)-c(\mu)
	\right\}.
	\label{eq:med-diagnostic-objective}
\end{equation}
This extension endogenizes information acquisition while retaining the baseline assumption of posterior-optimal treatment.

To see the diagnostic incentive, consider $\mu^H,\mu^L\in\mathcal E$ with $\mu^H\succeq_B\mu^L$, where $\succeq_B$ denotes Blackwell dominance. The patient's medical outcome is weakly better under the more informative diagnosis since $\Delta I:=I_{\mu^H}(M)-I_{\mu^L}(M)\geq0$. Naturally, the more accurate diagnostic usually comes with a higher cost, that is, $\Delta c:=c(\mu^H)-c(\mu^L)\geq0.$

The net change for the physician to switch from $\mu^L$ to $\mu^H$ is thus $(1+K)\Delta I-\Delta c$, which is increasing in $K$. Thus, greater anticipated scrutiny strengthens the incentive to acquire more information. In particular, if $K_1>K_0$ and
\[
	(1+K_0)\Delta I
	<
	\Delta c
	<
	(1+K_1)\Delta I,
\]
the physician selects $\mu^H$ under $K_1$ but $\mu^L$ under $K_0$. Because $\Delta c>\Delta I$, the additional information does not cover its resource cost under the material criterion $I_\mu(M)-c(\mu)$. The reduction in testing following immunity is therefore consistent with the removal of the additional return that scrutiny assigns to diagnostic information.

We next compare blanket immunity with targeted protection. A clinical guideline alone is advice; a \emph{guideline-based malpractice safe harbor} is an ex-ante rule that legally protects qualifying guideline-concordant care \citep{BlumsteinEtAl2020,BlumsteinEtAl2023}. Let $B$ be the protected treatment menu and let $A\supseteq B$ also contain a potentially useful unprotected treatment. Choosing $B$ results in residual scrutiny $K_s$, whereas choosing $A$ leaves the physician at $K$, where $0<K_s<K$. This menu-dependent scrutiny level extends the baseline representation, so the comparison is across institutional regimes rather than between menus under a single fixed-$K$ preference.

Because $H(B)\geq I_\mu(B)$ for every diagnostic technology, lowering scrutiny weakly raises the optimized value of $B$: $\mathcal V_{K_s}(B)\geq\mathcal V_K(B)$. Consequently,
\begin{equation}
	\mathcal V_K(A)-\mathcal V_{K_s}(B)
	\leq
	\mathcal V_K(A)-\mathcal V_K(B).
	\label{eq:med-optimized-menu-effect}
\end{equation}
Targeted protection therefore weakly lowers the relative value of the larger, unprotected menu. The resulting trade-off is between defensive diagnosis and treatment flexibility: lower scrutiny within $B$ can reduce low-value diagnostic testing, but it also makes $B$ more attractive relative to a menu containing patient-specific alternatives. When an excluded alternative would improve material welfare, the policy may therefore replace excessive diagnosis with excessive standardization. This does not imply that safe harbors are undesirable; it suggests combining clear protection for guideline-concordant care with an exception for documented, clinically justified departures, rather than treating the guideline as an exhaustive menu.

\subsection{Failure Tolerance and Confidential Experimentation}
\label{subsec:finance-confidential-experimentation}

\citet{TianWang2014} find that, among VC-backed IPO firms, backing by a more failure-tolerant venture capitalist is associated with greater innovation, especially for ventures facing high failure risk. They provide evidence that selection of firms with greater ex-ante innovative potential does not fully explain the relationship and show that career concerns reduce failure tolerance. \citet{Manso2011} and \citet{EdererManso2013} emphasize a dynamic moral-hazard channel: rewarding long-run success while tolerating early failure encourages exploration. Our channel is complementary but distinct: it operates even under aligned material interests. An investor delegates project generation, screening, and implementation to a manager who shares the investor's material utility; anticipated hindsight scrutiny can nevertheless discourage the generation of novel projects.

The extension retains posterior-optimal implementation but allows the hindsight benchmark to depend on disclosure. Let $A$ be the finite project menu actually generated and screened by the manager, and let $D\subseteq A$ contain the projects whose identities and realized payoffs enter the official record. Assume that the posterior-optimal project is unique and denote it by
\[
	f_A(p)
	\in
	\argmax_{f\in A}
	\sum_{\omega\in\Omega}p(\omega)u(f(\omega)).
\]
The implemented project is always observable, whether or not it was previously disclosed. The manager's ex-ante choice is therefore a pair $(A,D)$: the actual project menu and its disclosure set. The corresponding \emph{visible scrutiny gap} is
\begin{align}
	G(A,D)
	:={}&
	\int_{\DeltaO}
	\sum_{\omega\in\Omega}p(\omega)
	\left[
	\max_{f\in D\cup\{f_A(p)\}}u(f(\omega))
	-u(f_A(p)(\omega))
	\right]\mu(dp).
	\label{eq:fin-visible-gap}
\end{align}
Thus, $D=A$ corresponds to full disclosure, whereas $D=\emptyset$ corresponds to the case where only the implemented project is disclosed. Let $I(A)$ denote the baseline informed-choice value of $A$. Let $b(A,D)$ collect the coordination, monitoring, and organizational-learning benefits of disclosure net of confidentiality costs. The manager evaluates $(A,D)$ according to
\begin{equation}
	V_K(A,D)
	=
	I(A)-K G(A,D)+b(A,D).
	\label{eq:fin-disclosure-objective}
\end{equation}

A simple example illustrates how disclosure can affect the actual project menu. There are two equally likely market states. A legacy project $f$ yields $u(f)=(4,4)$, whereas a novel project $g$ yields $u(g)=(8,0)$. The manager receives one of two equally likely signals, which induce posterior probabilities $3/4$ and $1/4$, respectively, on the first state. She implements $g$ after the favorable signal and $f$ after the unfavorable signal. Normalizing the common value of $b$ across the relevant regimes to zero gives
\[
	V_K(\{f,g\},\{f,g\})=5-K,
	\qquad
	V_K(\{f,g\},\{f\})=5-\frac{K}{2},
	\qquad
	V_K(\{f\},\{f\})=4.
\]

Under mandatory full disclosure, the broader menu is strictly preferred if and only if $5-K>4$, or equivalently, $K<1$. For $K>1$, the manager therefore excludes $g$ from consideration. Mandatory disclosure changes the actual project menu, rather than merely the disclosure record, when $1<K<2$. In this range, the manager would strictly prefer to consider $g$ while keeping it off the formal record unless it is implemented---that is, to choose $A=\{f,g\}$ and $D=\{f\}$---if confidentiality were permitted. For $K>2$, even confidential consideration of $g$ is unattractive.

More generally, fix $A$ and $D\subsetneq A$, and take $e\in A\setminus D$. Define $\Delta G_e:=G(A,D\cup\{e\})-G(A,D)$ and $\Delta b_e:=b(A,D\cup\{e\})-b(A,D)$. Suppose that additional disclosure yields weakly greater organizational benefits, so that $\Delta b_e\geq0$. Because adding $e$ weakly raises the statewise maximum among visible projects, $\Delta G_e\geq0$. The manager strictly prefers to keep $e$ outside
the disclosure set if and only if
\[
	K\Delta G_e>\Delta b_e.
\]
If $\Delta G_e>0$, confidentiality is therefore strictly preferred when $K>\frac{\Delta b_e}{\Delta G_e}$. If adding $e$ never raises the statewise maximum among visible projects, then $\Delta G_e=0$. Anticipated scrutiny then provides no reason to conceal $e$.

Hence, scrutiny makes disclosure selective rather than uniformly undesirable. Because an implemented project is already visible, $\Delta G_e$ arises only when $e$ is rejected after private screening but would have appeared attractive in the realized state. Confidentiality can therefore protect a rejected experiment from becoming a hindsight benchmark; it cannot shield an implemented innovation from an unfavorable outcome. In this sense, confidential screening complements rather than replaces genuine tolerance for failure.

The model adds a disclosure margin to the findings of \citet{TianWang2014}. A lower $K$ can increase both the willingness to generate novel projects and the willingness to place them on the formal record. When scrutiny is high but confidential screening remains feasible, formal records may understate actual experimentation. When confidentiality is prohibited or too costly, the same pressure can contract the actual project menu. In the example, mandatory disclosure for $1<K<2$ changes behavior rather than merely revealing it: the manager switches from $\{f,g\}$ to $\{f\}$. The prediction is not that secrecy is generally optimal; the benefits of transparency are included in $b(A,D)$. The relevant design question is whether those benefits require every rejected experiment to become an individualized hindsight benchmark.

\section{Conclusion}
\label{sec:conclusion}

This paper develops a model of menu choice under hindsight scrutiny. Conservatism isolates the menu expansions for which more options are unambiguously beneficial: those that add flexibility without raising the state-by-state hindsight benchmark. Together with standard menu-preference axioms, conservatism characterizes a representation in which the material value of informed choice is weighed against hindsight scrutiny.

Menu-preference data identify a unique minimum information--intensity pair but generally do not recover the agent's actual information and scrutiny intensity separately. Combining menu-preference data with stochastic-choice data closes that gap: For every nontrivial jointly represented pair, the posterior distribution $\mu$ and scrutiny intensity $K$ are uniquely identified.

The applications show that accountability can reshape more than the final action. It can alter diagnostic information acquisition, the set of treatments or projects made available, and which experiments enter the formal record. The broader lesson is that accountability systems should be evaluated by the counterfactual benchmarks they create as well as by the conduct they deter. Extending the model to endogenous evaluation procedures, institution-dependent scrutiny, and dynamic reputational consequences offers a natural direction for future work.


\appendix

\section{Preliminaries}
\label{app:preliminaries}
In this section, we introduce auxiliary notation and preliminary results that are useful for our proofs.

For any set $Y$ and $y\in Y$, let $\delta_y$ denote the Dirac (probability) measure concentrated at $y$. For a measure $\xi$, $\supp(\xi)$ denotes its support, and we write $\R_+:=[0,\infty)$.

For an affine utility function $u:X\to\R$ and a menu $A\in\Menus$, define the support function $\sigma_A^u:\Delta(\Omega)\to\R$ by
\begin{equation}
	\label{eq:supportfunction}
	\sigma_A^u(p):=\max_{f\in A}\sum_{\omega\in\Omega}p(\omega)u(f(\omega)).
\end{equation} 
For a prior $\pi\in\Delta(\Omega)$, define the hindsight function $H_{u,\pi}:\Menus\to\R$ by
\begin{equation}
	\label{eq:hindsight}
	H_{u,\pi}(A):=\sum_{\omega\in\Omega}\pi(\omega)\sigma_A^u(\delta_\omega).
\end{equation}

For a scrutiny representation $(u,\mu,K)$, define the informed-choice function $I_{u,\mu}:\Menus\to\R$ by
\begin{equation}
	\label{eq:informed-choice}
	I_{u,\mu}(A):=\int_{\Delta(\Omega)}\sigma_A^u(p)\mu(dp),
\end{equation}
and the anticipated-scrutiny function $s_{u,\mu,K}:\Menus\to\R$ by
\begin{equation}
	\label{eq:anticipated-scrutiny}
	s_{u,\mu,K}(A):=K[H_{u,\pi}(A)-I_{u,\mu}(A)].
\end{equation}

When the representation is fixed, we abbreviate these functions as $I$, $H$, and $s$, respectively; when several representations are compared, a common subscript indexes all three. Thus, $I(A)$ is the informed-choice component of the value of menu $A$, $H(A)$ is its perfect-hindsight benchmark, and $s(A)$ is its ex-ante anticipated scrutiny. The scrutiny representation $(u,\mu,K)$ can be equivalently expressed as
\begin{equation}
	V(A)=(1+K)I(A)-KH(A)=I(A)-s(A).
\end{equation}

The following lemma lists some standard properties of support functions.
\begin{lemma}
	\label{lem:supportfunctionproperties}
	For any $A,B\in\Menus$ and any affine utility $u:X\to\R$:
	\begin{enumerate}
		\item If $A\subseteq B$, then $\sigma_A^u(p)\leq\sigma_B^u(p)$ for all $p\in\Delta(\Omega)$.
		\item $\sigma_{\alpha A+(1-\alpha)B}^u(p)=\alpha\sigma_A^u(p)+(1-\alpha)\sigma_B^u(p)$ for any $\alpha\in[0,1]$ and $p\in\Delta(\Omega)$.
	\end{enumerate}
\end{lemma}
Finally, for any $\pi\in\Delta(\Omega)$, define $\nu_\pi:=\sum_{\omega\in\Omega}\pi(\omega)\delta_{\delta_\omega}.$

\section{Proof of Theorem 1}
\label{app:proof-characterization}

We first record that a scrutiny representation belongs to the additive EU class.

\begin{lemma}
	\label{lem:scrutiny-additive}
	If $\succsim$ has a scrutiny representation $(u,\mu,K)$, then $\succsim$ has an additive EU representation $(u,\eta)$.
\end{lemma}
\begin{proof}
	Let $\pi:=\int_{\Delta(\Omega)}p\,\mu(dp)$ be the Bayesian prior corresponding to $\mu$, and recall that $\nu_\pi:=\sum_{\omega\in\Omega}\pi(\omega)\delta_{\delta_\omega}.$ Let $\eta:=(1+K)\mu-K\nu_\pi$. Then,
	\begin{align}
		&\ \int_{\Delta(\Omega)}\max_{f\in A}\left(\sum_{\omega\in\Omega}p(\omega)u(f(\omega))\right)\eta(dp) \\
		=&\ \int_{\Delta(\Omega)}\sigma_A^u(p)\big[(1+K)\mu-K\nu_\pi\big](dp) \\
		=&\ (1+K)\int_{\Delta(\Omega)}\sigma_A^u(p)\mu(dp)-K\int_{\Delta(\Omega)}\sigma_A^u(p)\nu_\pi(dp)\\
		=&\ (1+K)\int_{\Delta(\Omega)}\sigma_A^u(p)\mu(dp)-K\sum_{\omega\in\Omega}\pi(\omega)\max_{f\in A}u(f(\omega)) \\
		=&\ V(A).
	\end{align}
	Moreover, \(\eta(\Delta(\Omega))=(1+K)-K=1\), and its barycenter is
	$$\int_{\Delta(\Omega)}p\,\eta(dp)=(1+K)\pi-K\pi=\pi\in\text{int}(\Delta(\Omega)).$$
	Thus, \((u,\eta)\) is an additive EU representation for $\succsim$.
\end{proof}

The next lemma establishes the necessity of conservatism for the scrutiny representation.
\begin{lemma}
	\label{lem:scrutiny-defensible}
	If $\succsim$ has a scrutiny representation, then $\succsim$ satisfies conservatism.
\end{lemma}
\begin{proof}[Proof of \Cref{lem:scrutiny-defensible}]
	Suppose the scrutiny representation is $(u,\mu,K)$. Fix a menu $A$ and an act $g$ dominated by $A$. Let $\pi$ be the barycenter of $\mu$. The scrutiny representation can be rewritten as
	\begin{equation}
		V(A)=(1+K)\int_{\Delta(\Omega)}\sigma_A^u(p)\mu(dp)-K\sum_{\omega\in\Omega}\pi(\omega)\sigma_A^u(\delta_\omega).
	\end{equation}
	Since $g$ is dominated by $A$, 
    $$\sigma_A^u(\delta_\omega)=\max_{f\in A}u(f(\omega))=\max_{f\in A\cup\{g\}}u(f(\omega))=\sigma_{A\cup\{g\}}^u(\delta_\omega)\text{ for any }\omega\in\Omega.$$
	Therefore,
	$V(A\cup\{g\})-V(A) 
		=(1+K)\int_{\Delta(\Omega)}[\sigma_{A\cup\{g\}}^u(p)-\sigma_A^u(p)]\mu(dp)
		\geq 0$, 
	where the final inequality follows from part (a) of
	\Cref{lem:supportfunctionproperties}.
\end{proof}

The next lemma states that the axioms are sufficient if the preference is lottery-trivial (i.e., $x\sim y$ for any $x,y\in X$).

\begin{lemma}
	\label{lem:trivial-scrutiny}
	If $\succsim$ is lottery-trivial (i.e., $x\sim y$ for any $x,y\in X$) and satisfies \cref{ax:weak-order,ax:strongcontinuity,ax:independence,ax:domination,ax:PDF}, then $\succsim$ has a scrutiny representation with $u(x)=0$ for all $x\in X$.
\end{lemma}
\begin{proof}[Proof of \Cref{lem:trivial-scrutiny}]
	By \Cref{prop:additive-foundation}, \(\succsim\) has an additive EU representation. Because its restriction to \(X\) is lottery-trivial, the affine utility index in that representation is constant; normalize it to \(u(x)=0\) for every \(x\in X\). Every menu then has value zero. Hence \(\succsim\) is constant on \(\Menus\) and is represented by suitable \(\mu\) and $K$ together with this \(u\).
\end{proof}

\begin{lemma}
	\label{lem:support-difference}
	Suppose $u:X\to\R$ is nonconstant and affine. For every nonnegative continuous piecewise-affine function $h:\Delta(\Omega)\to\R_+$, there exist finite menus $A\supseteq B$ and a scalar $a>0$ such that
		$$\sigma_A^u(p)-\sigma_B^u(p)=ah(p)\quad\text{ for every }p\in\Delta(\Omega).$$
\end{lemma}
\begin{proof}[Proof of \Cref{lem:support-difference}]
	We use three standard facts about continuous piecewise-affine functions on a finite-dimensional polytope (e.g., $\Delta(\Omega)$):
	\begin{itemize}
		\item Fact 1. Every such function can be written as the difference of two convex continuous piecewise-affine functions \citep{KripfganzSchulze1987}.
		\item Fact 2. Every convex continuous piecewise-affine function is the maximum of finitely many affine functions \citep[Theorem~2.49]{RockafellarWets1998}.
		\item Fact 3. Every affine function on $\Delta(\Omega)$ can be written as $p\mapsto p\cdot v$ for some vector $v\in\R^\Omega$. (This follows from the fact that $\sum_\omega p(\omega)=1$ for any $p\in\Delta(\Omega)$.)
	\end{itemize}
	By Fact 1, we can write $h$ as $h=\phi-\psi$ where $\phi,\psi:\Delta(\Omega)\to\R$ are convex and piecewise-affine. Since $h\geq 0$, $\phi(p)\geq\psi(p)$ for every $p\in\Delta(\Omega)$. By Fact 2, there are finite families of affine functions $\{\ell_i\}_{i=1}^I$ and $\{m_j\}_{j=1}^J$ such that
	$$\phi(p)=\max_{1\leq i\leq I}\ell_i(p)\text{ and }\psi(p)=\max_{1\leq j\leq J}m_j(p)\text{ for every }p\in\Delta(\Omega).$$
	By Fact 3, there exist vectors $\{v_i\}_{i=1}^I$ and $\{w_j\}_{j=1}^J$ from $\R^\Omega$ such that
	$$\ell_i(p)=p\cdot v_i\text{ for each $i=1,\dots, I$ and }m_j(p)=p\cdot w_j\text{ for each $j=1,\dots, J$}.$$
	We next rescale these vectors so that they can be generated by actual acts. Define
	$\underline{u}:=\min_{x\in X}u(x)$ and $\bar u:=\max_{x\in X}u(x).$
	Because $u$ is nonconstant, $\underline u<\bar u$. Moreover, every number in $[\underline u,\bar u]$ is the utility of some lottery. Let
	$$r:=\min_{i,j,\omega}\{v_i(\omega),w_j(\omega)\}\text{ and } s:=\max_{i,j,\omega}\{v_i(\omega),w_j(\omega)\}.$$
	Because there are only finitely many vectors, $r$ and $s$ are finite. We can thus choose some $a>0$ such that $a(s-r)\leq\bar u-\underline u.$ By setting $b=\underline u-ar$, we can guarantee that every coordinate of every vector $av_i+b\mathbf{1}$ and $aw_j+b\mathbf{1}$ belongs to $[\underline u,\bar u]$.
	
	As a result, for every $i\in\{1,\dots, I\}$, there is an act $f_i$ such that $u(f_i(\omega))=av_i(\omega)+b\text{ for every }\omega\in\Omega.$ Similarly, for every $j\in\{1,\dots, J\}$, there is an act $g_j$ satisfying $u(g_j(\omega))=aw_j(\omega)+b\text{ for every }\omega\in\Omega.$ Define two finite menus by $B=\{g_1,\dots, g_J\}$ and $C=\{f_1,\dots, f_I\}.$ For any $p\in\Delta(\Omega)$,
		\begin{align*}
			\max_{g\in B}\sum_{\omega\in\Omega}p(\omega)u(g(\omega))&=\max_{1\leq j\leq J}\sum_{\omega\in\Omega}
			p(\omega)[aw_j(\omega)+b]=\max_{1\leq j\leq J}\{am_j(p)+b\}
			=a\psi(p)+b.
		\end{align*}
	Thus, $\sigma_B^u=a\psi+b$. Similarly, $\sigma_C^u=a\phi+b$. Let $A=B\cup C$. Then $B\subseteq A$ and
	$$\sigma_A^u(p)=\max\{\sigma_B^u(p),\sigma_C^u(p)\}=\max\{a\psi(p)+b,a\phi(p)+b\}=a\phi(p)+b,$$
	where the last equality follows from $a>0$ and $\phi\geq\psi$. By construction,
	$$\sigma_A^u(p)-\sigma_B^u(p)=a[\phi(p)-\psi(p)]=ah(p)\text{ for every }p\in\Delta(\Omega).$$
	This completes the proof of \Cref{lem:support-difference}.
\end{proof}
For any nonempty \(S\subseteq\Omega\), let \(D(S):=\{\delta_\omega\mid\omega\in S\}\) denote the set of degenerate beliefs concentrated on states in $S$. Let \(\mathcal H_S\) be the collection of nonnegative continuous piecewise-affine functions \(h:\Delta(\Omega)\to\R_+\) satisfying $h(\delta_\omega)=0$ for every $\omega\in S$.

\Cref{lem:positive-off-degenerate,lem:null-state-support} present intermediate results that are more general than needed to deal with cases where the prior $\pi$ does not have full support. For our result under the full-support-prior assumption, it suffices to consider $S=\Omega$.
\begin{lemma}
	\label{lem:positive-off-degenerate}
	Let \(S\) be a nonempty subset of \(\Omega\), and let $\eta$ be a finite signed Borel measure on $\Delta(\Omega)$. Suppose
	\begin{equation}
		\int_{\Delta(\Omega)}h(p)\eta(dp)\geq 0
		\text{ for every }h\in\mathcal H_S,
		\label{eq:positive-HS}
	\end{equation}
	then the restriction of $\eta$ to $\Delta(\Omega)\setminus D(S)$ is a positive measure.
\end{lemma}
\begin{proof}[Proof of \Cref{lem:positive-off-degenerate}]
	We want to show: $\eta(A)\geq 0$ for every Borel set $A\subseteq\Delta(\Omega)\setminus D(S)$.
	
	We extend the hypothesis from $h\in\mathcal H_S$ to arbitrary nonnegative continuous functions vanishing on $D(S)$. Fix a continuous function $g:\Delta(\Omega)\to\R_+$ satisfying \(g(\delta_\omega)=0\) for every \(\omega\in S\). We want to show from equation~\eqref{eq:positive-HS} that $\int_{\Delta(\Omega)}g(p)\eta(dp)\geq 0.$
	
	We do this by approximating $g$ using a sequence of functions $g_n\in\mathcal H_S$. Let $\mathcal T_n$ be the $n$th iterate of barycentric subdivisions of $\Delta(\Omega)$. Iterated barycentric subdivision is a triangulation of $\Delta(\Omega)$ that guarantees: (i) every $\delta_\omega$ is a vertex of $\mathcal T_n$ for every $n$; and (ii) the largest diameter of a simplex in $\mathcal T_n$ converges to zero.
	
	To construct $g_n$, let $m:=|\Omega|-1$ and define $g_n$ to agree with $g$ at every vertex of $\mathcal T_n$ and to be affine on each simplex of $\mathcal T_n$. Equivalently, if a simplex has vertices $q_{n,0},q_{n,1},\dots,q_{n,m}$ and $p$ has barycentric coordinates $\{\lambda_{n,k}(p)\}_{k=0}^m$ in that simplex, then
	\begin{equation}
			g_n(p):=\sum_{k=0}^m\lambda_{n,k}(p)g(q_{n,k}).
	\end{equation}
	The affine interpolants agree on shared faces because they take the same values at every shared vertex. Hence $g_n$ is globally well defined and continuous. By construction, $g_n(p)\geq 0$ for every $p\in\Delta(\Omega)$. Moreover, because each $\delta_\omega$ is a vertex in $\mathcal T_n$, $g_n(\delta_\omega)=g(\delta_\omega)=0$ for every \(\omega\in S\). In sum, $g_n\in\mathcal H_S$ for each $n$.
	
	We now show that $g_n\to g$ uniformly. For any $p\in\Delta(\Omega)$,
	\begin{equation}
			|g_n(p)-g(p)|=\left|\sum_{k=0}^m\lambda_{n,k}(p)[g(q_{n,k})-g(p)]\right|\leq\sum_{k=0}^m\lambda_{n,k}(p)|g(q_{n,k})-g(p)|.
	\end{equation}
	Let $\delta_n$ denote the largest diameter of a simplex in $\mathcal T_n$ (such a max exists because there are always only finitely many simplices in $\mathcal T_n$). By construction, $\delta_n\to 0$. Each vertex $q_{n,k}$ is within distance $\delta_n$ of $p$. Since $g$ is continuous on the compact set $\Delta(\Omega)$, $g$ must also be uniformly continuous (Heine--Cantor theorem). Thus, for any $\varepsilon>0$, there exists $\delta>0$ such that ``$\lVert p-q\rVert<\delta\Longrightarrow |g(p)-g(q)|<\varepsilon\text{ for all }p,q\in\Delta(\Omega).$'' Since $\delta_n\to 0$, there exists $N$ such that $n\geq N\Longrightarrow \delta_n<\delta$, which implies that
	$$|g_n(p)-g(p)|\leq \sum_{k=0}^m\lambda_{n,k}(p)|g(q_{n,k})-g(p)|<\varepsilon\text{ for every }p\in\Delta(\Omega).$$
	That is, $g_n$ converges to $g$ uniformly.
	
	Our final step is to extend the hypothesis to any nonnegative continuous $g$ that vanishes on $D(S)$.
	
	Because $\eta$ is a finite signed Borel measure,
	\begin{equation}
			\left|\int_{\Delta(\Omega)}[g_n(p)-g(p)]\eta(dp)\right|\leq\Vert g_n-g\Vert_{\infty}|\eta|(\Delta(\Omega)),
	\end{equation}
	where $|\eta|$ is the total-variation measure. $|\eta|(\Delta(\Omega))$ is finite by assumption and we have shown $g_n\to g$ uniformly. Thus, $\Vert g_n-g\Vert_{\infty}|\eta|(\Delta(\Omega))\to 0$, which implies that
	$$\int_{\Delta(\Omega)}g_n(p)\eta(dp)\to\int_{\Delta(\Omega)}g(p)\eta(dp).$$
	By the hypothesis of the Lemma, $\int_{\Delta(\Omega)}g_n(p)\eta(dp)\geq 0\text{ for every $n$.}$ Taking the limit preserves the inequality and thus $\int_{\Delta(\Omega)}g(p)\eta(dp)\geq 0.$
	
	Fix any $\varphi:\Delta(\Omega)\setminus D(S)\to\R_+$ that is continuous with compact support contained in $\Delta(\Omega)\setminus D(S)$ where the support of $\varphi$ is defined by $\supp(\varphi):=\overline{\{p\in\Delta(\Omega)\setminus D(S)\mid \varphi(p)\neq 0\}}$. Define an extension of $\varphi$ to $\Delta(\Omega)$ by
	\begin{equation}
		g(p)=\begin{cases}\varphi(p),&\text{ if }p\in\Delta(\Omega)\setminus D(S), \\ 0, & \text{ if }p\in D(S)\end{cases}.
	\end{equation}
	Since $\supp(\varphi)$ is compact and disjoint from $D(S)$, there exists an open neighborhood of $D(S)$ that does not intersect $\supp(\varphi)$. On this entire neighborhood, $\varphi(p)=0$. Thus, $g$ is identically zero in a neighborhood of every element of \(D(S)\). This makes it continuous at each $\delta_\omega\in D(S)$, and thus continuous on $\Delta(\Omega)$.
	
	Therefore, the constructed $g$ is nonnegative, continuous and vanishes on $D(S)$. Thus,
	$$\int_{\Delta(\Omega)\setminus D(S)}\varphi(p)\eta(dp)=\int_{\Delta(\Omega)}g(p)\eta(dp)\geq 0.$$
	
	Notice that our choice of $\varphi$ is arbitrary. Therefore, we have shown that ``the restriction of $\eta$ to $\Delta(\Omega)\setminus D(S)$ integrates every nonnegative continuous and compactly supported function to a nonnegative number.'' By the Riesz representation theorem for the dual of continuous functions on a locally compact Hausdorff space \citep[Theorem~14.12]{AliprantisBorder2006}, together with uniqueness of the representing regular signed measure, the quoted result implies that $\eta(A)\geq 0$ for every Borel set $A\subseteq\Delta(\Omega)\setminus D(S)$.
\end{proof}

\begin{lemma}
	\label{lem:null-state-support}
	Let \(\eta\) be a finite signed Borel measure on \(\Delta(\Omega)\) whose barycenter \(\pi:=\int_{\Delta(\Omega)}p\,\eta(dp)\) belongs to \(\Delta(\Omega)\), and let \(S=\supp(\pi)\). If the restriction of \(\eta\) to \(\Delta(\Omega)\setminus D(S)\) is positive, then \(\eta\) is supported on \(\Delta(S)\).
\end{lemma}
\begin{proof}[Proof of \Cref{lem:null-state-support}]
	Fix \(\omega\notin S\). Since \(\pi(\omega)=0\), $0=\int_{\Delta(\Omega)}p(\omega)\eta(dp).$ Every \(\delta_s\in D(S)\) assigns probability zero to \(\omega\). Consequently, $0=\int_{\Delta(\Omega)\setminus D(S)}p(\omega)\eta(dp).$
	
	The restriction of \(\eta\) in this integral is a positive measure and the integrand is nonnegative. It follows that this restriction assigns zero mass to \(\{p:p(\omega)>0\}\).
	
	Since $\Delta(S)=\bigcap_{\omega\notin S}\{p\mid p(\omega)=0\}$, repeating the above argument for every \(\omega\notin S\) shows that \(\eta\) assigns zero mass outside $\Delta(S)$. The possible atoms in \(D(S)\) already lie in this face. Hence \(\eta\) is supported on \(\Delta(S)\).
\end{proof}

We are now ready to prove \Cref{thm:characterization}.

\begin{proof}[Proof of \Cref{thm:characterization}]
	The necessity of Axioms 1 to 4 follows from \Cref{lem:scrutiny-additive} and \Cref{prop:additive-foundation}; necessity of Axiom 5 follows from \Cref{lem:scrutiny-defensible}. 
	
	The sufficiency of the axioms when \(\succsim\) is lottery-trivial follows from \Cref{lem:trivial-scrutiny}. Suppose henceforth that \(\succsim\) is lottery-nontrivial. By \Cref{prop:additive-foundation}, \(\succsim\) has an additive EU representation \((u,\eta)\), where \(u\) is nonconstant, $\eta(\Delta(\Omega))=1$ with its barycenter $\pi:=\int_{\Delta(\Omega)}p\,\eta(dp)$ in the interior of $\Delta(\Omega)$ such that $\succsim$ is represented by $W(A)=\int_{\Delta(\Omega)}\sigma_A^u(p)\eta(dp).$ 
	
	Let \(S:=\supp(\pi)=\Omega\). The rest of the proof only relies on $S=\supp(\pi)$. 
	
	We first use conservatism to locate the negative part of \(\eta\). Fix \(h\in\mathcal H_S\). By \Cref{lem:support-difference}, there are finite menus \(A\supseteq B\) and \(a>0\) such that $\sigma_A^u(p)-\sigma_B^u(p)=ah(p)$ for every $p\in\Delta(\Omega)$. Because \(h(\delta_\omega)=0\) for every \(\omega\in S\),
	\[
	\begin{aligned}
		H_{u,\pi}(A)-H_{u,\pi}(B)
		&=\sum_{\omega\in\Omega}\pi(\omega)
		[\sigma_A^u(\delta_\omega)-\sigma_B^u(\delta_\omega)]
		=a\sum_{\omega\in S}\pi(\omega)h(\delta_\omega)=0.
	\end{aligned}
	\]
	Because $B\subseteq A$, each difference $\sigma_A^u(\delta_\omega)-\sigma_B^u(\delta_\omega)$ is nonnegative. Since $\pi$ has full support, the preceding equality implies $\sigma_A^u(\delta_\omega)=\sigma_B^u(\delta_\omega)$ for every $\omega\in\Omega$. Hence every $g\in A\setminus B$ is dominated by $B$: for each state $\omega$, an act in the finite menu $B$ that attains $\sigma_B^u(\delta_\omega)$ is weakly preferred to $g(\omega)$. If $A=B$, then $A\sim B$. Otherwise, list $A\setminus B$ as $\{g_1,g_2,\dots,g_n\}$. Applying conservatism and transitivity repeatedly gives
	$$A=B\cup\{g_1,g_2,\dots,g_n\}\succsim B\cup\{g_2,\dots, g_n\}\succsim\cdots\succsim B\cup\{g_n\}\succsim B.$$
	Thus, $A\succsim B$. Thus, $a\int_{\Delta(\Omega)}h(p)\eta(dp)=W(A)-W(B)\geq 0$. Since our choice of \(h\in\mathcal H_S\) was arbitrary, \Cref{lem:positive-off-degenerate} implies that the restriction of \(\eta\) to \(\Delta(\Omega)\setminus D(S)\) is positive.
	
	This positivity also eliminates any hidden dependence on null states. By \Cref{lem:null-state-support}, the barycenter condition implies that \(\eta\) is supported on \(\Delta(S)\). Thus the only possible negative mass consists of the finitely many atoms \(\eta(\{\delta_\omega\})\) with \(\omega\in S\). In particular, $\eta(\{\delta_\omega\})=0$ for every $\omega\notin S$.
	
	We now obtain the scrutiny representation through a reparametrization. 
	
	Because \(\pi(\omega)>0\) for every \(\omega\in S\), we can choose \(K\geq0\) sufficiently large that $\eta(\{\delta_\omega\})+K\pi(\omega)\geq0$ for every $\omega\in S$. Because $S=\supp(\pi)=\Omega$, the full-revelation distribution $\nu_\pi$ is supported on $D(S)$. Define $\mu:=\frac{1}{1+K}\eta+\frac{K}{1+K}\nu_\pi.$ Then for every Borel set \(A\subseteq\Delta(\Omega)\), $(1+K)\mu(A)=\eta(A\setminus D(S))+\sum_{\omega\in S:\,\delta_\omega\in A}[\eta(\{\delta_\omega\})+K\pi(\omega)]\geq0.$ Hence \(\mu\) is positive. Its total mass is one, and $\int_{\Delta(\Omega)}p\,\mu(dp)=\frac{1}{1+K}\pi+\frac{K}{1+K}\pi =\pi.$ Finally, substituting $\eta=(1+K)\mu-K\nu_\pi$ into the additive EU representation gives
	\[
	\begin{aligned}
		W(A)
		&=(1+K)\int_{\Delta(\Omega)}\sigma_A^u(p)\mu(dp)
		-K\int_{\Delta(\Omega)}\sigma_A^u(p)\nu_\pi(dp)\\
		&=(1+K)\int_{\Delta(\Omega)}\sigma_A^u(p)\mu(dp)
		-K\sum_{\omega\in\Omega}\pi(\omega)
		\max_{f\in A}u(f(\omega)),
	\end{aligned}
	\]
	which is the reduced form of a scrutiny representation.
\end{proof}

\section{Other Omitted Proofs}

\subsection{Proof of \Cref{lem:lottery-nontriviality-equivalence}}
\label{app:proof-lottery-nontriviality-equivalence}
\begin{proof}
	Let $(u,\eta)$ be an additive EU representation, so that $W(A)=\int_{\Delta(\Omega)}\sigma_A^u(p)\,\eta(dp)$ and $\eta(\Delta(\Omega))=1.$ Lottery nontriviality plainly implies that $\succsim$ is nonconstant.
	
	Conversely, suppose lottery nontriviality fails. Since $W(\{x\})=\int_{\Delta(\Omega)}u(x)\,\eta(dp)=u(x)$ for every $x\in X$, failure of lottery nontriviality implies that $u$ is constant on $X$, say $u\equiv c$. Consequently, $\sigma_A^u(p)=c$ for every $A\in\Menus$ and every $p\in\Delta(\Omega)$, and hence $W(A)=c$ for every $A\in\Menus$. Thus $\succsim$ is constant. Taking the contrapositive completes the proof.
\end{proof}
\subsection{Proof of \Cref{lem:zero-scrutiny}}
\label{app:proof-zero-scrutiny}
\begin{proof}
	Suppose first that $(u,\mu,0)$ represents $\succsim$. If $A\supseteq B$, then $\sigma_A^u(p)\geq\sigma_B^u(p)$ for every $p\in\Delta(\Omega)$. Therefore, $V(A)-V(B)=\int_{\Delta(\Omega)}\bigl[\sigma_A^u(p)-\sigma_B^u(p)\bigr]\,\mu(dp)\geq0,$ so $\succsim$ satisfies preference for flexibility.
	
	Conversely, suppose that $\succsim$ satisfies preference for flexibility, and fix a scrutiny representation $(u,\mu,K)$. Let $\pi$ be the barycenter of $\mu$, let $\nu_\pi:=\sum_{\omega\in\Omega}\pi(\omega)\delta_{\delta_\omega}$ and $\eta:=(1+K)\mu-K\nu_\pi.$ Then $V(A)=\int_{\Delta(\Omega)}\sigma_A^u(p)\,\eta(dp),$ while $\eta$ has total mass one and barycenter $\pi$. If $u$ is constant, $(u,\mu,0)$ already represents the same constant preference. Suppose, therefore, that $u$ is nonconstant. By Lemma \ref{lem:support-difference}, for every nonnegative continuous piecewise-affine function $h:\Delta(\Omega)\to\R_+$, there exist finite menus $A\supseteq B$ and $a>0$ such that $\sigma_A^u-\sigma_B^u=ah.$ Preference for flexibility therefore gives $\int_{\Delta(\Omega)}h(p)\,\eta(dp)=a^{-1}[V(A)-V(B)]\geq 0$.  Nonnegative continuous piecewise-affine functions are uniformly dense in the set of nonnegative continuous functions on $\Delta(\Omega)$. Since $\eta$ is finite, it follows that $\int_{\Delta(\Omega)}g(p)\,\eta(dp)\geq0$ for every continuous $g\geq0$. Hence $\eta$ is a positive measure. Because it has total mass one and barycenter $\pi$,  $(u,\eta,0)$ is a valid scrutiny representation of $\succsim$.
\end{proof}

\subsection{Proof of \Cref{lem:minimum-pair}}
\label{app:proof-minimum-pair}
\begin{proof}
	For any prior $\pi$, recall that $\nu_\pi$ denotes the posterior distribution induced by full revelation. Every scrutiny representation $(u,\mu,K)$ induces the additive signed measure $\eta=(1+K)\mu-K\nu_\pi,$ with $\pi=\int_{\Delta(\Omega)}p\,\eta(dp).$ Because $\succsim$ is nonconstant, \Cref{prop:additive-uniqueness} implies that $\eta$, and therefore $\pi$, is common to all scrutiny representations of $\succsim$.
	
	The proof of \Cref{thm:characterization} shows that $\eta$ is nonnegative away from the degenerate posteriors. Moreover, if $\pi(\omega)=0$, Bayes plausibility implies $\mu(\{\delta_\omega\})=0$, and hence $\eta(\{\delta_\omega\})=0$. Define $\underline K:=\max_{\omega:\,\pi(\omega)>0}\left(-\dfrac{\eta(\{\delta_\omega\})}{\pi(\omega)}\right)_+$, where $r_+:=\max\{r,0\}$. Define $\underline\mu:=\dfrac{\eta+\underline K\nu_\pi}{1+\underline K}$. By construction, $\underline\mu$ is a nonnegative measure. It has total mass one and barycenter $\pi$. Thus $\underline\mu$ is Bayes-plausible with respect to $\pi$, and $\eta=(1+\underline K)\underline\mu-\underline K\nu_\pi.$ Consequently, $(\underline\mu,\underline K)\in\mathcal P(\succsim)$.
	
	Now fix $(\mu,K)\in\mathcal P(\succsim)$. Uniqueness of $\eta$ gives $\mu=\frac{1}{1+K}\eta+\frac{K}{1+K}\nu_\pi.$ Nonnegativity of $\mu$ at each $\delta_\omega$ with $\pi(\omega)>0$ implies $\eta(\{\delta_\omega\})+K\pi(\omega)\geq0$ for each $\omega\in\Omega$. Therefore $K\geq\underline K$. Substituting the definition of
	$\underline\mu$ yields
	\[
	\mu=\frac{1+\underline K}{1+K}\,\underline\mu+\frac{K-\underline K}{1+K}\,\nu_\pi.
	\]
	Because full revelation Blackwell dominates every Bayes-plausible posterior distribution, this mixture is Blackwell more informative than $\underline\mu$.
	
	Finally, if $\underline K=0$, then $(u,\underline\mu,0)$ is a scrutiny representation, \Cref{lem:zero-scrutiny} then implies $\succsim$ satisfies preference for flexibility. Conversely, if $\succsim$ satisfies preference for flexibility, \Cref{lem:zero-scrutiny} implies that it has a scrutiny representation with $K=0$. Since $\underline K\geq0$ by definition, this forces $\underline K=0$. 
\end{proof}

\subsection{Proof of \Cref{thm:menu-uniqueness}}
\label{app:proof-menu-uniqueness}
\begin{proof}
	``Only if.'' Suppose two minimum representations induce the same lottery-nontrivial preference. They induce an additive signed measure $\eta=(1+K)\mu-K\nu_\pi$, where $\pi$ is the barycenter of $\mu$. Since $\succsim$ is lottery-nontrivial, \Cref{prop:additive-uniqueness} implies that $u\approx u'$ and $\eta=\eta'$.  In particular, both measures have the same barycenter $\pi$.
	
	The construction in \Cref{lem:minimum-pair} determines the minimum pair uniquely from $\eta$ and $\pi$:
	\[
	\underline K=\max_{\omega:\,\pi(\omega)>0}\left(-\frac{\eta(\{\delta_\omega\})}{\pi(\omega)}
	\right)_+,
	\qquad
	\underline\mu=\frac{\eta+\underline K\nu_\pi}{1+\underline K}.
	\]
	Because both representations are minimum representations, it follows that $K=K'=\underline K$ and $\mu=\mu'=\underline\mu$.
	
	``If.'' Conversely, suppose that $u\approx u'$, $\mu=\mu'$, and $K=K'$. Write $u'=au+b$ for some $a>0$ and $b\in\R$. The common posterior distribution implies that the functions generated by the two utilities satisfy $I'(A)=aI(A)+b$ and $H'(A)=aH(A)+b.$ Consequently, $V'(A)=(1+K)I'(A)-KH'(A)=aV(A)+b.$ Thus \(V\) and \(V'\) induce the same preference.
\end{proof}

\subsection{Proof of \Cref{thm:comparative-scrutiny}}
\label{app:proof-comparative-scrutiny}
\begin{proof}
	First suppose (2) holds. Without loss, normalize $u_1=u_2=:u$. $\mu_1=\mu_2$ then implies that $\pi_1=\pi_2:=\pi$, and hence the functions $I_1=I_2:=I$ and $H_1=H_2:=H$. Define $\Gamma(A):=H(A)-I(A)$ to be the common hindsight gap. $\Gamma\geq 0$.  Since $V_i(A)=I(A)-K_i\Gamma(A)$ and $V_i(\{f\})=\sum_{\omega\in\Omega}\pi(\omega)u(f(\omega))$ does not depend on $K_i$, $K_1\geq K_2$ implies condition (a). If $A$ and $B$ dominate each other under both preferences, their common utility index implies that $\max_{f\in A}u(f(\omega))=\max_{g\in B}u(g(\omega))$ for every $\omega\in\Omega$. Hence $H(A)=H(B)$ and $V_i(A)-V_i(B)=(1+K_i)[I(A)-I(B)]$, so condition (b) also holds.
	
	Conversely, suppose (1) holds (i.e., conditions (a) and (b) hold).  For an act $f$, let
	\[
	v_i(f):=V_i(\{f\})=\sum_{\omega\in\Omega}\pi_i(\omega)u_i(f(\omega)),\ i=1,2.
	\]
	Applying condition (a) to singleton menus shows that $v_2(f)\geq v_2(g)$ implies $v_1(f)\geq v_1(g)$.  Equality under $v_2$ implies equality under $v_1$ by applying this implication in both directions. Because both $v_1,v_2:\Acts\ \to\R$ are nonconstant and affine, it follows that $v_1=av_2+b$ for some $a>0$ and $b\in\R$.  Restriction to constant acts gives $u_1=au_2+b$.  After normalizing the utilities to be equal, variation of an act in one state at a time gives $\pi_1=\pi_2=: \pi$.
	
	We now show that $\mu_1=\mu_2$. 
	
	Since $\pi_1=\pi_2=\pi$, it is without loss to normalize $u_1=u_2:=u$ with $u(X)=[0,1]$. Then $H_1(A)=H_2(A)$ for any $A\in\Menus$. Let the common hindsight value be $H$. Let $\Gamma_i(A):=H(A)-I_i(A)$ denote the expected hindsight gap under $\mu_i$. Notice that $\Gamma_i$ is affine with respect to menu mixtures and that $\Gamma_i(\{f\})=0$ for every act $f$.
	
	We first show that $\Gamma_1$ and $\Gamma_2$ induce the same ordering over all menus. Fix arbitrary $A,B\in\Menus$, and write $m_C(\omega):=\max_{h\in C}u(h(\omega))$ for any menu $C$. Choose $\alpha>0$ sufficiently small that $r(\omega):=\frac{\alpha}{1-\alpha}[m_B(\omega)-m_A(\omega)]\in[-1,1]$ for every $\omega$. Because $u(X)=[0,1]$, there exist acts $f,g$ such that $u(f(\omega))-u(g(\omega))=r(\omega)$ for every $\omega$. Define $\widehat A:=\alpha A+(1-\alpha)\{f\}$ and $\widehat B:=\alpha B+(1-\alpha)\{g\}$.

	The two new menus have the same statewise envelope: $m_{\widehat A}(\omega)=m_{\widehat B}(\omega)$ for every $\omega$. They therefore dominate each other under both $\succsim_1$ and $\succsim_2$. Since $\Gamma_i$ vanishes on singleton menus, $\Gamma_i(\widehat A)=\alpha \Gamma_i(A)$ and $\Gamma_i(\widehat B)=\alpha \Gamma_i(B)$ for $i=1,2$. Thus,
	\[
	V_i(\widehat A)-V_i(\widehat B)
	=(1+K_i)[I_i(\widehat A)-I_i(\widehat B)]=
	-(1+K_i)
	[\Gamma_i(\widehat A)-\Gamma_i(\widehat B)].
	\]
	Applying condition (b) to $\widehat A$ and $\widehat B$ gives: $\Gamma_1(A)\leq \Gamma_1(B)$ if and only if $\Gamma_2(A)\leq \Gamma_2(B).$
	
	Thus, $\Gamma_1,\Gamma_2:\Menus\to\R$ are nonconstant affine representations of the same ordering over menus. Affine-utility uniqueness implies that there is some $a>0$ and $b\in\R$ such that $\Gamma_1=a\Gamma_2+b.$ Since both $\Gamma_1$ and $\Gamma_2$ vanish on every singleton menu, $b=0$. Hence $\Gamma_1=a\Gamma_2$ for some $a>0$. Because $I_i=H-\Gamma_i$, $\Gamma_1=a\Gamma_2$ implies that $I_1=aI_2+(1-a)H.$
	
	Recall that $\nu_\pi$ is the full-revelation distribution. Then $H(A)=\int_{\Delta(\Omega)}\sigma_A^u(p)\nu_\pi(dp)$. 
	
	Suppose, toward a contradiction, that $a<1$. Then $I_1=aI_2+(1-a)H$ implies
	\[
	\int_{\Delta(\Omega)}\sigma_A^u(p)\mu_1(dp)=\int_{\Delta(\Omega)}\sigma_A^u(p)\bigl[a\mu_2+(1-a)\nu_\pi\bigr](dp).
	\]
	Both $\mu_1$ and $a\mu_2+(1-a)\nu_\pi$ are probability measures giving subjective-learning representations of the preference induced by $I_1$. The uniqueness result of Dillenberger et al.\ (2014, Theorem~1) therefore implies $\mu_1=a\mu_2+(1-a)\nu_\pi.$ Because $(u,\mu_1,K_1)$ is a minimum representation for a flexibility-violating $\succsim_1$, the maximum defining $K_1$ in \Cref{lem:minimum-pair} binds at some $\omega\in\supp(\pi)$. Consequently, $\mu_1(\{\delta_\omega\})=0$ for at least one such $\omega$. But  $\mu_1=a\mu_2+(1-a)\nu_\pi$ would imply $\mu_1(\{\delta_\omega\})\geq(1-a)\pi(\omega)>0$ for every $\omega\in\supp(\pi)$, a contradiction. Therefore $a<1$ is impossible.
	
	Suppose, toward a contradiction, that $a>1$. Then $\Gamma_2=a^{-1}\Gamma_1$ and $I_2=a^{-1}I_1+(1-a^{-1})H.$ Applying the same identification result gives $\mu_2=a^{-1}\mu_1+(1-a^{-1})\nu_\pi,$ which contradicts the fact that the minimum distribution $\mu_2$ must have zero mass at some $\delta_\omega$ with $\omega\in\supp(\pi)$. Therefore $a>1$ is also impossible.
	
	Thus $a=1$, so $I_1=I_2$. Applying the uniqueness result of DLST again yields $\mu_1=\mu_2$.
	
	It remains to order the scrutiny parameters.  
	
	Since the preferences represented by the minimum representations are flexibility-violating, $\mu\neq\nu_\pi$. Thus $\mu$ assigns positive probability to nondegenerate posteriors; the menu consisting of the state-indicator acts then has $\Gamma(A):=H(A)-I(A)>0$.  If $K_1<K_2$, then $V_1(A)>V_2(A)$.  Choose a constant act $x$ with utility in the interior of $u(X)$ and, for small $\alpha>0$, set $A_\alpha:=\alpha A+(1-\alpha)\{x\}$.  Both $V_i(A_\alpha)$ then lie in the interior of $u(X)$ and $V_1(A_\alpha)>V_2(A_\alpha)$.  A constant act $y$ can be chosen with utility strictly between these two values.  Thus $\{y\}\succ_2 A_\alpha$ but $A_\alpha\succ_1\{y\}$, contradicting condition (a).  Hence $K_1\geq K_2$, completing the proof.
\end{proof}

\subsection{Proof of \Cref{prop:binary-choice-identification}}
\label{app:proof-binary-choice-identification}
\begin{proof}
	We first identify utility from choices between constant acts. Suppose $x,y\in X$ and $u(x)>u(y)$. Then $x$ is the unique maximizer from $\{x,y\}$ at every posterior, so $\rho_{\{x,y\}}(x)=1$ under the first representation. Under the second representation this rules out $u'(y)>u'(x)$; hence $u'(x)\geq u'(y)$. Interchanging the two representations gives
	\begin{equation}
		u(x)>u(y)\Longrightarrow u'(x)\geq u'(y)\quad\text{ and }\quad u'(x)>u'(y)\Longrightarrow u(x)\geq u(y).
		\label{eq:weak-order-preservation}
	\end{equation}
	We claim that $u'$ is constant on every level set of $u$. Suppose $u(x)=u(y)=u_0$. Suppose by contradiction (and without loss of generality) that $u'(x)>u'(y)$. If $u_0\neq \min u(X)$, choose $z$ with $u(z)<u(x)$ and set $x_\varepsilon=(1-\varepsilon)x+\varepsilon z$. For all sufficiently small $\varepsilon>0$, we have $u'(x_\varepsilon)>u'(y)$ but $u(x_\varepsilon)<u(y)$, contradicting \eqref{eq:weak-order-preservation}. If $u_0=\min u(X)$, choose $z$ with $u(z)>u(y)$ and instead perturb $y$ toward $z$; the same argument gives a contradiction. Thus $u'$ is constant on the level sets of $u$. Since both functions are affine and nonconstant on the convex set $X$, it follows that $u'\approx u$.
	
	Normalize $u'$ and $u$ so that $u(X)=[0,1]$ and $u'=u$. For every vector $v=(v_\omega)_{\omega\in\Omega}\in[0,1]^\Omega$, define the act $f_v$ by $f_v(\omega):=x_{v_\omega}.$ From the binary menu $\{f_v,x_t\}$, the support restriction on tie-breaking gives
	\begin{align}
		&\ \mu\{p\mid p\cdot v>t\}\leq \rho_{\{f_v,x_t\}}(f_v)\leq \mu\{p\mid p\cdot v\geq t\}; \\
		&\ \mu'\{p\mid p\cdot v>t\}\leq \rho_{\{f_v,x_t\}}(f_v)\leq \mu'\{p\mid p\cdot v\geq t\}.
		\label{eq:projection-sandwich}
	\end{align}
	At every $t$ that is a continuity point of the distributions of $p\cdot v$ under both $\mu$ and $\mu'$, both inequalities collapse to equality. The common binary choice probability therefore implies that $p\cdot v$ has the same distribution under $\mu$ and $\mu'$: the common continuity points are dense, and a distribution function is determined by its values on any dense set of continuity points. This holds for every $v\in[0,1]^\Omega$.
	
	We now use the following finite-dimensional identification result. The \emph{Cram\'er--Wold theorem} states that if $\nu_1$ and $\nu_2$ are Borel probability measures on $\R^m$ and the scalar projection $x\mapsto a\cdot x$ has the same distribution under $\nu_1$ and $\nu_2$ for every $a\in\R^m$, then $\nu_1=\nu_2$. Equivalently, the distribution of a random vector is uniquely determined by the distributions of all its linear combinations \citep[pp.~290--294]{CramerWold1936}.
	
	To apply the theorem, regard each posterior $p$ as a vector in $\R^{|\Omega|}$. Every $a\in\R^\Omega$ can be written as $a=c\mathbf 1+r v$, where $v\in[0,1]^\Omega$, $c\in\R$ and $r\geq 0$. Indeed, if $a$ is nonconstant, take $c=\min_{\omega}a(\omega)$, $r=\max_{\omega}a(\omega)-\min_{\omega}a(\omega)$ and $v=\frac{a-c\mathbf 1}{r}$. If $a$ is constant, take $r=0$. Because $p\cdot\mathbf 1=1$ for every posterior $p$, $p\cdot a=c+r(p\cdot v).$ We have already shown that $p\cdot v$ has the same distribution under $\mu$ and $\mu'$ for every $v\in[0,1]^\Omega$. Hence $p\cdot a$ has the same distribution under the two measures for every $a\in\R^\Omega$. The Cram\'er--Wold theorem therefore implies $\mu=\mu'$.
\end{proof}

\subsection{Proof of \Cref{lem:recover-choice-value}}
\label{app:proof-choice-value}
\begin{proof}
	Conditional on $p$, every posterior maximizer belongs to $A$ if
	$\sigma_A^u(p)>t$, while $x_t$ is the unique maximizer if $\sigma_A^u(p)<t$. If $\sigma_A^u(p)=t$, the tie-breaking kernel may assign any probability to choosing from $A$. Consequently,
	\[
	\mu\{p\mid\sigma_A^u(p)>t\}\leq T_A^\rho(t)\leq \mu\{p\mid\sigma_A^u(p)\geq t\}.
	\]
	For completeness, Tonelli's theorem states that if $(S,\alpha)$ and $(T,\beta)$ are $\sigma$-finite measure spaces and $h:S\times T\to[0,\infty]$ is measurable, then the order of integration may be reversed:
	\begin{equation}
		\int_S\!\left(\int_T h(s,t)\,\beta(dt)\right)\alpha(ds)
		=
		\int_T\!\left(\int_S h(s,t)\,\alpha(ds)\right)\beta(dt),
		\label{eq:tonelli-statement}
	\end{equation}
	where the common value may be $+\infty$ \citep[Theorem~2.37]{Folland1999}. In our application, $\alpha=\mu$, $\beta$ is Lebesgue measure on $[0,1]$, and
	\[
	h_{>}(p,t):=\mathbf 1\{t<\sigma_A^u(p)\}, \qquad h_{\geq}(p,t):=\mathbf 1\{t\leq\sigma_A^u(p)\}.
	\]
	These functions satisfy Tonelli's hypotheses. They are nonnegative. Moreover, because $A$ is finite, $\sigma_A^u$ is the maximum of finitely many continuous affine functions of $p$ and is therefore continuous. Thus the sets $\{(p,t)\mid t<\sigma_A^u(p)\}$ and $\{(p,t)\mid t\leq\sigma_A^u(p)\}$ are Borel measurable, so their indicators are measurable. Both measures are finite and hence $\sigma$-finite.
	
	Since $u(X)=[0,1]$, we also have $0\leq\sigma_A^u(p)\leq1$. Applying Tonelli to $h_{>}$ gives
	\begin{align*}
		\int_0^1\mu\{p\mid\sigma_A^u(p)>t\}\,dt
		&=\int_0^1\int_{\Delta(\Omega)}
		\mathbf 1\{t<\sigma_A^u(p)\}\,\mu(dp)\,dt\\
		&=\int_{\Delta(\Omega)}\int_0^1
		\mathbf 1\{t<\sigma_A^u(p)\}\,dt\,\mu(dp)=\int_{\Delta(\Omega)}\sigma_A^u(p)\,\mu(dp).
	\end{align*}
	The last equality holds because, for fixed $p$, the inner integrand is the indicator of $[0,\sigma_A^u(p))$, whose Lebesgue measure is $\sigma_A^u(p)$. Applying the same argument to $h_{\geq}$ gives
	\[
	\int_0^1\mu\{p\mid\sigma_A^u(p)\geq t\}\,dt=\int_{\Delta(\Omega)}\sigma_A^u(p)\,\mu(dp),
	\]
	because $[0,\sigma_A^u(p)]$ has the same Lebesgue measure as $[0,\sigma_A^u(p))$. The integral of $T_A^\rho$ is sandwiched between these two equal quantities, which proves \eqref{eq:choice-value} and also shows why tie-breaking at $t=\sigma_A^u(p)$ does not affect the result.
\end{proof}

\subsection{Proof of \Cref{thm:joint-uniqueness}}
\label{app:proof-joint-uniqueness}
\begin{proof}
	\Cref{prop:binary-choice-identification} gives $u'\approx u$ and $\mu'=\mu$, and the common posterior distribution has a common Bayesian prior $\pi$. Normalize the common utility so that $u(X)=[0,1]$, and write as before
	\[
	I(A)=\int_{\Delta(\Omega)}\sigma_A^u(p)\,\mu(dp), \qquad H(A)=\sum_{\omega}\pi(\omega)\max_{f\in A}u(f(\omega)), \qquad \Gamma(A)=H(A)-I(A).
	\]
	The two normalized menu indices are $V(A)=I(A)-K\Gamma(A)$ and $V'(A)=I(A)-K'\Gamma(A)$. Both $V$ and $V'$ are affine representations of the same flexibility-violating preference on the mixture space $\Menus$. Hence $V'=aV+b$ for some $a>0$ and $b\in\R$. For every constant test act $x_t$, we have $V(\{x_t\})=V'(\{x_t\})=t.$ Therefore $t=at+b$ for every $t\in[0,1]$, which implies $a=1$ and $b=0$. It follows that $(K'-K)\Gamma(A)=0$ for every $A\in\Menus$. By construction, $\Gamma(A)\geq 0$ for every $A$. Moreover, flexibility violation by $\succsim$ implies that $\Gamma(B)>0$ for some $B\in\Menus$. Putting these together, it must be $K'=K$.
\end{proof}

\subsection{Proof of \Cref{thm:joint-characterization}}
\label{app:proof-joint-characterization}
\begin{proof}
	By \Cref{prop:binary-choice-identification}, the normalized utility index and the distribution of posterior beliefs are uniquely identified by $\rho$. Fix an information representation $(u,\mu,\tau)$ of $\rho$ and normalize $u$ so that $u(X)=[0,1]$. Let $\pi$ denote the barycenter of $\mu$, so that $\pi(\omega):=\int_{\Delta(\Omega)}p(\omega)\,\mu(dp)$ for every $\omega\in\Omega$. For every menu $A\in\Menus$, define its \emph{hindsight benchmark} by
	\[
		H_\rho(A)
		:=
		\sum_{\omega\in\Omega}\pi(\omega)
		\max_{f\in A}u(f(\omega)).
	\]
	This definition is independent of the particular information representation chosen because the normalized $u$ and $\mu$, and hence their barycenter $\pi$, are identified by $\rho$.

	``Only if.'' Suppose $(\succsim,\rho)$ has a joint scrutiny representation, with scrutiny intensity $K$. Under the normalization fixed above, its menu-value functional can be written as
	\begin{equation}
		V(A)
		=
		(1+K)I_\rho(A)-K H_\rho(A).
		\label{eq:joint-value}
	\end{equation}
	If the premises of JSC hold, then $I_\rho(A)-I_\rho(B)\geq 0$. Because $B$ dominates $A$, $H_\rho(B)\geq H_\rho(A)$. Equation~\eqref{eq:joint-value} therefore gives $V(A)\geq V(B)$, which implies $A\succsim B$.
	
	``If.'' Suppose $(\succsim,\rho)$ satisfies SA and JSC, and continue to use the information representation $(u,\mu)$ of $\rho$ fixed above. We want to show that there exists $K\geq 0$ such that $(u,\mu,K)$ is a scrutiny representation of $\succsim$.
	
	Recall that $\pi$ is the barycenter of $\mu$. Let $W':\Menus\to\R$ be a scrutiny representation of $\succsim$ with parameters $(u',\mu',K')$. $I_\rho(\{f\})=\sum_{\omega\in\Omega}\pi(\omega)u(f(\omega))$ for any $f\in\Acts$. SA dictates that $W'$ restricted to singleton menus induces the same EU preference over acts as the one induced by $I_\rho$ restricted to singleton menus. Thus, the standard Anscombe--Aumann result guarantees that $W:\Menus\to\R$ with parameters $(u,\mu',K')$ is also a scrutiny representation of $\succsim$ and the barycenter of $\mu'$ must also be $\pi$.
	
	\noindent\textbf{Claim}: If $I_\rho(A)\geq I_\rho(B)$ and $H_\rho(A)\leq H_\rho(B)$, then $W(A)\geq W(B)$.
	
	Proof of the claim: Recall that $u(X)=[0,1]$, and write $m_C(\omega):=\max_{h\in C}u(h(\omega))$ for any menu $C$. Fix menus $A,B$ satisfying the claim's premises. Choose $\alpha\in(0,1)$ sufficiently small, and let $r(\omega):=\frac{\alpha}{1-\alpha}[m_A(\omega)-m_B(\omega)]$, $c:=-\sum_{\omega}\pi(\omega)r(\omega)$, and $d(\omega):=r(\omega)+c$. Since $H_\rho(A)\leq H_\rho(B)$, we have $c\geq0$. Moreover, $\sum_{\omega}\pi(\omega)d(\omega)=0$, and $\alpha$ can be chosen small enough that $d(\omega)\in[-1,1]$ for every $\omega$.

	Because $u(X)=[0,1]$, there exist acts $f,g$ such that $u(g(\omega))-u(f(\omega))=d(\omega)$ for every $\omega$. Set $A_\alpha:=\alpha A+(1-\alpha)\{f\}$ and $B_\alpha:=\alpha B+(1-\alpha)\{g\}$. The statewise envelope of $B_\alpha$ exceeds that of $A_\alpha$ by $(1-\alpha)c$ in every state, so $B_\alpha$ dominates $A_\alpha$. Also, the zero prior mean of $d$ implies $I_\rho(\{f\})=I_\rho(\{g\})$, and hence $I_\rho(A_\alpha)-I_\rho(B_\alpha)=\alpha[I_\rho(A)-I_\rho(B)]\geq0$. JSC therefore gives $A_\alpha\succsim B_\alpha$. Finally, $W$ is affine and $W(\{f\})=W(\{g\})$, so $W(A_\alpha)\geq W(B_\alpha)$ implies $W(A)\geq W(B)$, proving the claim.

	In particular, if $I_\rho(A)=I_\rho(B)$ and $H_\rho(A)=H_\rho(B)$, applying the claim in both directions gives $W(A)=W(B)$. Thus $W(A)$ depends on $A$ only through $I_\rho(A)$ and $H_\rho(A)$.
	
	Returning to the main proof, write $I=I_\rho$ and $H=H_\rho$. We want to show that there exists $K\geq 0$ such that $W(A)=(1+K)I(A)-KH(A)$ is a scrutiny representation with parameters $(u,\mu,K)$, where $I$ and $H$ are the corresponding informed-choice value and hindsight benchmark.
	
	Let $T:\Menus\to\R^2$ be the two-dimensional map $T(A):=(I(A),H(A))$, and let $\mathcal R:=\{(I(A),H(A))\mid A\in\Menus\}=T(\Menus)$. The set $\mathcal R$ is convex: if $(i_1,h_1),(i_2,h_2)\in\mathcal R$, then there exist $A_1,A_2\in\Menus$ with $(i_1,h_1)=T(A_1)$ and $(i_2,h_2)=T(A_2)$, and the linearity of $I$ and $H$ guarantees that $\alpha(i_1,h_1)+(1-\alpha)(i_2,h_2)=T(\alpha A_1+(1-\alpha) A_2)\in\mathcal R$ for every $\alpha\in[0,1]$.
	
	By the claim, we can define $\widetilde W:\mathcal R\to\R$ by $\widetilde W(i,h):=W(A)$ for any $A$ satisfying $T(A)=(i,h)$; the claim guarantees that the definition does not depend on which such menu $A$ is chosen. Moreover, arguments similar to those above show that $\widetilde W$ is affine on the convex set $\mathcal R$; that is,
	\[
		\widetilde W(\lambda r_1+(1-\lambda)r_2)
		=
		\lambda \widetilde W(r_1)
		+
		(1-\lambda)\widetilde W(r_2),
	\]
	for any $r_1,r_2\in\mathcal R$ and $\lambda\in[0,1]$.
	
	\noindent\textbf{Case 1:} $I(A)=H(A)$ for every $A\in\Menus$. 
		
	Fix any $A$, let $t:=I(A)=H(A)$. For every constant act $x$, $W(\{x\})=I(\{x\})=H(\{x\})=u(x)$. Since $\succsim$ is lottery-nontrivial, $u(X)$ is not a singleton and $t\in u(X)$. There exists $x_t\in X$ such that $u(x_t)=t$, $I(A)=I(\{x_t\})$, and $H(A)=H(\{x_t\})$. The claim applied in both directions gives $A\sim\{x_t\}$. Therefore, $W(A)=W(\{x_t\})=I(A)$. Setting $K=0$ guarantees that $(u,\mu,0)$ represents $\succsim$.
	\\[4pt]
	\noindent\textbf{Case 2:} $I(A)\neq H(A)$ for some $A\in\Menus$. 
	
	Note that $I(\{f\})=H(\{f\})$ for every singleton menu, and the existence of $A$ with $I(A)\neq H(A)$ implies that the affine hull of $\mathcal R$ is the entire $\R^2$. Then the affinity of $\widetilde W$ implies that there exists $a,b,c\in\R$ such that 
	\[
		\widetilde W(i,h)=ai+bh+c,
	\]
	which is equivalent to saying that 
	\[
		W(A)=aI(A)+bH(A)+c
		\quad\text{for every menu }A\in\Menus.
	\]
	We then show that JSC implies $a\geq 0$ and $b\leq 0$.
		
	Because $\mathcal R$ is full-dimensional, it contains an interior point $(i,h)$. Hence, for sufficiently small $\varepsilon>0$, the points $(i+\varepsilon,h)$ and $(i,h+\varepsilon)$ also belong to $\mathcal R$.
		
	To see $a\geq 0$, choose menus $A,B$ such that
	\[
		(I(A),H(A))=(i+\varepsilon,h)
		\quad\text{and}\quad
		(I(B),H(B))=(i,h).
	\]
	Then $I(A)>I(B)$ and $H(A)=H(B)$. The claim implies $A\succsim B$, so $0\leq W(A)-W(B)=a\varepsilon$.
		
	To see $b\leq 0$, choose menus $A,B$ with $(I(A),H(A))=(i,h)$ and $(I(B),H(B))=(i,h+\varepsilon)$. The claim gives $A\succsim B$, so $0\leq W(A)-W(B)=-b\varepsilon$.
		
	For every constant act $x$, $W(\{x\})=I(\{x\})=H(\{x\})=u(x)$. It then follows from the above that $c=0$ and $a+b=1$. Let $K:=-b\geq0$. Then $a=1+K$, so
	\[
		W(A)
		=
		(1+K)I(A)-K H(A)
		=
		(1+K)I_\rho(A)-KH_\rho(A),
	\]
	is the desired scrutiny representation of $\succsim$ with parameters $(u,\mu,K)$.
		
	The tie-breaking kernels from the information representation of $\rho$ can be retained unchanged, completing the joint representation.

	This completes the proof.
\end{proof}



\BeginSupplementaryAppendix

	\begin{center}
		{\LARGE Supplementary Appendix}
	\end{center}
	
    \begin{center}
    {\small\textbf{Abstract}}
    \end{center}
    \begin{quote}
		{\small In this appendix, we provide a detailed proof of the additive expected-utility (EU) representation for preferences over menus of Anscombe--Aumann acts, building on the additive-representation results of \citet[henceforth DLR]{DLR2001} and \citet[henceforth DLRS]{DLRS2007}. The proof first establishes a slightly stronger result. The construction follows the geometric translation in \citet{DLST2014}: Utility acts are embedded into a simplex of lotteries, the corrected additive theorem of DLR and DLRS is applied there, and the result is translated back by downward closure and a weighted pushforward. We also prove the uniqueness of the additive representation.}
	\end{quote}
	
	\section{Environment}
	
	Let $\Omega$ and $Z$ be finite nonempty sets, with $m:=|\Omega|$.  Let $X:=\Delta(Z)$ be the lottery space and let $\Acts:=X^\Omega$ be the set of Anscombe--Aumann acts, endowed with its Euclidean metric.  Let $\Menus$ be the collection of nonempty compact subsets of $\Acts$, endowed with the Hausdorff metric $d_h$.  Mixtures of acts are defined statewise and mixtures of menus by 
	\[
	\alpha A+(1-\alpha)B:=\{\alpha f+(1-\alpha)g:f\in A,\ g\in B\}.
	\]
	Our primitive is a binary relation $\succsim$ over $\Menus$.  We employ the common abuse of notation as in the main text.
	
	\begin{axiom}[Weak Order]\label{ax:wo}
		The relation $\succsim$ is complete and transitive.
	\end{axiom}
	
	\begin{axiom}[Strong Continuity]\label{ax:sc}\
		\begin{enumerate}[label=(\roman*)]
		\item vNM Continuity: If $A\succ B\succ C$, then there are $\alpha,\bar\alpha\in(0,1)$ such that
		\[
		\alpha A+(1-\alpha)C\succ B\succ \bar\alpha A+(1-\bar\alpha)C.
		\]
		\item L Continuity: There exist menus $A^*,A_*$ and $\gamma>0$ such that, whenever $d_h(A,B)\leq\alpha/\gamma$ for some $\alpha\in(0,1)$,
		\[
		(1-\alpha)A+\alpha A^*\succsim(1-\alpha)B+\alpha A_*.
		\]
		\end{enumerate}
	\end{axiom}
	
	\begin{axiom}[Independence]\label{ax:ind}
		If $A\succ B$, then, for every menu $C$ and $\alpha\in(0,1]$,
		\[
		\alpha A+(1-\alpha)C\succ\alpha B+(1-\alpha)C.
		\]
	\end{axiom}

    Axioms S1 to S3 are exactly the same as axioms 1 to 3 from the main text. In the following, we decompose Axiom 4 in the main text into three separate axioms to isolate the effect of each one of them.
    
	\begin{axiom}[Domination]\label{ax:dom}
		If $f(\omega)\succsim g(\omega)$ for every $\omega\in\Omega$ and $f\in A$, then $A\sim A\cup\{g\}.$
	\end{axiom}
    This corresponds to clause (b) of Axiom 4 in the main text.
    
	\begin{axiom}[Act Monotonicity]\label{ax:monotonicity}
		If $f(\omega)\succsim g(\omega)$ for every $\omega\in\Omega$, then $f\succsim g$.
	\end{axiom}
    This corresponds to clause (a) of Axiom 4 in the main text.
    
    \begin{axiom}[Strict Act Monotonicity]\label{ax:strict-monotonicity}
	If $f(\omega)\succsim g(\omega)$ for every $\omega\in\Omega$ and there exists $\omega^*$ such that $f(\omega^*)\succ g(\omega^*)$, then $f\succ g$.
	\end{axiom}
    This corresponds to clause (c) of Axiom 4 in the main text.
	
	\begin{definition}\label{def:normalized-additive}
	A pair $(u,\eta)$ is a \emph{normalized additive representation} if $u:X\to\R$ is affine, $\eta$ is a finite countably additive signed Borel measure on $\DeltaO$ such that
	\[
	W(A)=\int_{\DeltaO}\left[\max_{f\in A}\sum_{\omega\in\Omega}p(\omega)u(f(\omega))\right]\,\eta(dp) \tag{NA}\label{eq:NA}
	\]
	represents $\succsim$ with $\eta(\DeltaO)=1$.
	\end{definition}
	Note that the affine utility over lotteries is allowed to be constant, and that no sign restriction is placed on the barycenter $\int_{\DeltaO}p\,\eta(dp)$. In particular, the barycenter can put negative weights on some states in $\Omega$. 
	
	\begin{definition}[Additive EU representation]\label{def:additive-EU}
	A normalized additive representation $(u,\eta)$ is an \emph{additive EU representation} if, in addition, $\int_{\DeltaO}p\,\eta(dp)\in\DeltaO.$ 
	\end{definition}
	That is, an additive EU representation is a normalized additive representation whose measure $\eta$ has a barycenter that is a probability measure over $\Omega$.
	
	\section{Representation Theorem and Proof}
	\begin{proposition}\label{prop:normalized-foundation}
	A binary relation $\succsim$ over $\Menus$ has a normalized additive representation if and only if it satisfies Axioms S1-S4.
	\end{proposition}
	Proposition S1 is slightly stronger than Proposition 1 in the main text.
	\medskip
	
	\noindent\textbf{Corollary.} \textit{Under Definition \ref{def:additive-EU}, Axioms S1--S5 characterize additive EU representations. Adding Axiom S6 yields the full-support version used in the main text. Proposition 1 therefore follows from Axioms S1--S6.}
    \begin{proof}[Proof of Corollary]
By \Cref{prop:normalized-foundation}, Axioms S1--S4 are equivalent to the existence of a normalized additive representation $(u,\eta)$. Let $\pi=\int_{\DeltaO}p\,\eta(dp)$. Since $\eta(\DeltaO)=1$, $\sum_{\omega\in\Omega}\pi(\omega)=1.$ Moreover, for acts $f$ and $g$, $W(f)-W(g)=\sum_{\omega\in\Omega}\pi(\omega)\bigl[u(f(\omega))-u(g(\omega))\bigr]$.

Suppose first that $u$ is nonconstant. Fix $\omega\in\Omega$, choose lotteries $x\succ y$, and consider two acts that coincide outside $\omega$ and yield $x$ and $y$, respectively, at $\omega$. Act Monotonicity implies that the first act is weakly preferred, so $\pi(\omega)\bigl[u(x)-u(y)\bigr]\geq0$. Hence $\pi(\omega)\geq0$. Because $\omega$ was arbitrary, $\pi\in\DeltaO$, and $(u,\eta)$ is an additive EU representation. Under Strict Act Monotonicity, the same comparison is strict, implying $\pi(\omega)>0$ for every $\omega$.

Conversely, if $\pi\in\DeltaO$, $\pi(\omega)\bigl[u(x)-u(y)\bigr]\geq0$ immediately implies Act Monotonicity; if $\pi$ has full support, it also implies Strict Act Monotonicity. Finally, if $u$ is constant, the preference is trivial and satisfies both axioms; it admits the required full-support representation by taking $\eta=\delta_{p_0}$ for any full-support $p_0\in\DeltaO$.
\end{proof}

	To facilitate the proof of \Cref{prop:normalized-foundation}, we introduce the following notion of continuity:
	\begin{axiom}[Hausdorff Continuity]\label{ax:hausdorff-continuity}
	For any $A\in\Menus$, the sets $\{B\in\Menus\mid A\succsim B\}$ and $\{B\in\Menus\mid B\succsim A\}$ are closed (in the Hausdorff topology).
	\end{axiom}
	
	We first isolate a standard implication of \Cref{ax:wo,ax:sc,ax:ind}.
	
	\begin{lemma}\label{lem:affine-lipschitz}
	If $\succsim$ satisfies Weak Order, Strong Continuity, and Independence, then $\succsim$ satisfies Hausdorff Continuity.
	\end{lemma}
	
	\begin{proof}The proof builds on Lemmas S6 to S8 in the online supplement to \citet[henceforth DLRS]{DLRS2007}. Although DLRS state these results for menus of lotteries, their arguments use only the finite-dimensional convex mixture structure and therefore apply to the present Anscombe--Aumann act space.
		
	Lemma S6 in the online supplement to DLRS gives $A\sim\co A$ for every $A\in\Menus$ where $\co A$ is the convex hull of $A$. On the collection of convex menus, Lemma S7 in that supplement gives an affine representation $V$ under Weak Order, vNM Continuity, and Independence.  Because Strong Continuity includes L-Continuity, Lemma S8 in the supplement (equivalently, Lemma 1 in the main corrigendum) implies that $V$ is Hausdorff-Lipschitz on that domain.  
		
    Extend $V$ to all menus by setting $V(A):=V(\co A)$.  The preceding indifference ensures that the extension still represents $\succsim$.  As noted in Section S2.2 of the supplement, convexification is nonexpansive in the Hausdorff metric:
	\[
	d_h(\co A,\co B)\leq d_h(A,B).
	\]
	Consequently, the extension is also Hausdorff-Lipschitz and hence continuous. Its upper and lower contour sets are therefore closed, which is Hausdorff Continuity.
	\end{proof}
	
	Now we prove Proposition S1 in the steps described in the abstract.
	
	\begin{proof}[Proof of Proposition S1.]
	\ \\
	\emph{``If.''} Suppose first that the restriction of $\succsim$ to constant acts (lotteries) is trivial.  Then, for any acts $f,g$, $f(\omega)\sim g(\omega)$ for every $\omega$. Domination implies that any finite menu is indifferent to each of its singleton submenus.  Approximate an arbitrary menu by finite submenus and use the Hausdorff continuity supplied by \Cref{lem:affine-lipschitz}.  It follows that all menus are indifferent. The preference can be represented by $u\equiv0$ and, for example, $\eta=\delta_{p_0}$ for any $p_0\in\DeltaO$.
		
	Henceforth suppose that the restriction of $\succsim$ to constant acts is nontrivial.
		
	\noindent
	\emph{\bf Step 1: From lottery consequences to utility acts.}
		
	The restriction of $\succsim$ to constant singleton acts satisfies the vNM axioms. Hence there is a nonconstant affine $u:X\to\R$ representing this restriction. It is without loss to normalize $u$ so that $u(X)=[0,1]$. Choose $x^0,x^1\in X$ with $u(x^0)=0$ and $u(x^1)=1$, and define $x^t:=(1-t)x^0+tx^1$. Then $u(x^t)=t$.
		
	Let $\Omega=\{\omega_1,\omega_2,\dots,\omega_m\}$.  For an act $f$, define its utility vector as
		\[
		T(f):=\bigl(u(f(\omega_1)),u(f(\omega_2)),\dots,u(f(\omega_m))\bigr)\in[0,1]^m.
		\]
	Conversely, for $v\in[0,1]^m$, write $v=(v_1,v_2,\dots,v_m)$ and define its canonical \emph{lift} $\widehat v\in\Acts$ by
		\[
		\widehat v(\omega_i):=x^{v_i}.
		\]
	Thus $T(\widehat v)=v$. It is straightforward to verify that the lift $v\mapsto\widehat v$ is affine, continuous, and Lipschitz. Let $\Menus^u$ denote the collection of nonempty compact subsets of $[0,1]^m$. For a utility menu $C$ in $\Menus^u$, write $\widehat C:=\{\widehat v:v\in C\}$ and define a binary relation $\succsim^u$ over $\Menus^u$ by $C\succsim^u D$ if $\widehat C\succsim\widehat D$. This binary relation is well defined because the lift $C\mapsto\widehat C$ is fixed.
		
	\begin{lemma}
		\label{lem:lift-equivalence}
		For a menu $A\in\Menus$, let $T(A):=\{T(f)\mid f\in A\}$. Then $A\sim\widehat{T(A)}.$
	\end{lemma}
		
	\begin{proof}[Proof of Lemma \ref{lem:lift-equivalence}]
	Indeed, every $f\in A$ is statewise indifferent to $\widehat{T(f)}$. Starting from $A$, add a countable dense collection from $\widehat{T(A)}$, one act at a time. Domination gives indifference at every finite step and Hausdorff continuity gives $A\sim A\cup\widehat{T(A)}$.  Starting from $\widehat{T(A)}$ and adding a dense collection from $A$ proves the reverse indifference through the same union.
	\end{proof}
		
	We now verify directly that $\succsim^u$ inherits all four axioms. That is:
	\begin{lemma}
		\label{lem:succuaxioms}
		If $\succsim$ satisfies Weak Order, Strong Continuity, Independence, and Domination, then $\succsim^u$ satisfies each one of those axioms adapted to utility menus.
	\end{lemma}
	\begin{proof}[Proof of Lemma \ref{lem:succuaxioms}]
	The binary relation $\succsim^u$ is complete and transitive because it is the pullback of $\succsim$ through the map $C\mapsto\widehat C$.  The lift preserves mixtures exactly, that is, for any utility menus $C,D\in\Menus^u$:
	\[
	\widehat{\alpha C+(1-\alpha)D}=\alpha\widehat C+(1-\alpha)\widehat D.
	\]
	Consequently, any strict comparison among lifted menus remains strict after a common mixture by Independence, which proves Independence for $\succsim^u$. The same identity translates the two mixtures supplied by vNM Continuity for $\succsim$ into the required mixtures of menus in $\Menus^u$, proving vNM Continuity for $\succsim^u$.
			
	It remains to verify L-Continuity. Endow $\Acts=X^\Omega$ with the product Euclidean metric and $[0,1]^m$ with its usual Euclidean metric, denoting both by $d$, and put $c:=d(x^1,x^0)>0$. Because $x^t=(1-t)x^0+tx^1$, for all $v,w\in[0,1]^m$, $d(\widehat v,\widehat w)=c\cdot d(v,w)$, and $d_h(\widehat C,\widehat D)=c\cdot d_h(C,D)$.
			
	Let $A^*,A_*$ and $\gamma$ witness L-Continuity for $\succsim$, and set $C^*:=T(A^*)$ and $C_*:=T(A_*)$.  Lemma~\ref{lem:lift-equivalence} gives $A^*\sim\widehat{C^*}$ and $A_*\sim\widehat{C_*}$. If $d_h(C,D)\leq\alpha/(c\gamma)$, then $d_h(\widehat C,\widehat D)\leq\alpha/\gamma$, so the original axiom gives
	\[
    (1-\alpha)\widehat C+\alpha A^*\succsim(1-\alpha)\widehat D+\alpha A_*.
	\]
	Preservation of indifference under common mixtures permits $A^*,A_*$ to be replaced by $\widehat{C^*},\widehat{C_*}$. Using the mixture identity above and the definition of $\succsim^u$, we obtain
	\[
	(1-\alpha)C+\alpha C^*\succsim^u(1-\alpha)D+\alpha C_*.
	\]
	Thus $C^*,C_*$ and $c\gamma$ witness L-Continuity for $\succsim^u$.
			
	Finally, $\succsim^u$ satisfies coordinatewise Domination:
		\begin{equation}
			v\geq w,\quad v\in C\quad\Longrightarrow\quad C\sim^u C\cup\{w\}.
			\label{eq:utility-domination}
		\end{equation}
	Indeed, $v\geq w$ implies $\widehat v(\omega)\succsim\widehat w(\omega)$ in every state because $u$ represents preferences over constant acts; the conclusion is therefore exactly \Cref{ax:dom} applied to the lifted menu.
	\end{proof}
		
	\noindent
	\emph{\bf Step 2: Embed utility acts into a lottery simplex.}
		
	\noindent Introduce an artificial prize, $m+1$, and define $J:[0,1]^m\to\Delta(\{1,\ldots,m+1\})$ by
	\begin{equation}
		J(v):=\left(\frac{v_1}{m},\ldots,\frac{v_m}{m},1-\frac1m\sum_{i=1}^m v_i\right).
	\label{eq:DLST-embedding}
	\end{equation}
	The final coordinate is nonnegative because $\sum_i v_i\leq m$ for any $v\in[0,1]^m$. Thus, $J(v)$ is a lottery over $m+1$ prizes.  Let
	\[
		P:=J([0,1]^m)=\left\{q\in\Delta(\{1,\ldots,m+1\})\mid q_i\leq\frac1m\text{ for }i\leq m\right\}.
	\]
	$J$ is an affine bijection between $[0,1]^m$ and $P$, with inverse $J^{-1}(q)=(mq_1,\ldots,mq_m)$.  Let $\Menus^P$ denote the set of all nonempty compact subsets of $P$. Define a binary relation $\succsim^P$ on $\Menus^P$ by $M\succsim^P N$ if $J^{-1}(M)\succsim^u J^{-1}(N)$.
		
	We verify the properties of this transported preference that will be used below. $\succsim^P$ is complete and transitive since $\succsim^u$ is complete and transitive.  Since $J^{-1}$ is affine,
	\[
		J^{-1}\bigl(\alpha M+(1-\alpha)N\bigr)=\alpha J^{-1}(M)+(1-\alpha)J^{-1}(N),
    \]
	so vNM Continuity and Independence pass from $\succsim^u$ to $\succsim^P$. For L-Continuity, observe directly from $J^{-1}(q)=(mq_1,\ldots,mq_m)$ that
	\[
        d_h\bigl(J^{-1}(M),J^{-1}(N)\bigr)\leq m\,d_h(M,N).
	\]
	Let $C^*,C_*$ and $\gamma_u$ witness L-Continuity for $\succsim^u$.  If $d_h(M,N)\leq\alpha/(m\gamma_u)$, then
	\[
		d_h\bigl(J^{-1}(M),J^{-1}(N)\bigr)\leq\frac{\alpha}{\gamma_u}.
    \]
	Applying L-Continuity for $\succsim^u$ and using the displayed mixture identity shows that $J(C^*)$ and $J(C_*)$, with constant $m\gamma_u$, witness L-Continuity for $\succsim^P$.
		
	Let $\Menus^{\ell}$ be the set of all nonempty compact subsets of $\Delta(\{1,\ldots,m+1\})$. We extend $\succsim^P$ from $P$ to $\Menus^\ell$.  Let $q^0:=\left(\frac1{m+1},\ldots,\frac1{m+1}\right)$ and fix $\varepsilon\in(0,1/m^2)$.  For every lottery $q$ and every $i\leq m$,
	\[
		\varepsilon q_i+(1-\varepsilon)\frac1{m+1}\leq\varepsilon+(1-\varepsilon)\frac1{m+1}\leq\frac1m.
	\]
	Consequently, if we fix any $M\in\Menus^\ell$ and let $S_\varepsilon(M):=\varepsilon M+(1-\varepsilon)\{q^0\}$, then $S_\varepsilon(M)\subseteq P$. Finally, define a binary relation $\succsim^\ell$ by $M\succsim^\ell N$ if $S_\varepsilon(M)\succsim^P S_\varepsilon(N).$ The definition does not depend on the particular sufficiently small $\varepsilon$. Indeed, if $0<\varepsilon'<\varepsilon$, then
	\[
		S_{\varepsilon'}(M)=\frac{\varepsilon'}{\varepsilon}S_\varepsilon(M)+\left(1-\frac{\varepsilon'}{\varepsilon}\right)\{q^0\},
	\]
	so Independence makes the rankings generated by $S_{\varepsilon'}$ and $S_\varepsilon$ equivalent.
		
	The extension agrees with $\succsim^P$ on its original domain $P$.  Take menus $M,N\subseteq P$.  If $M\succ^P N$, Independence gives
	\[
		S_\varepsilon(M)=\varepsilon M+(1-\varepsilon)\{q^0\}\succ^P\varepsilon N+(1-\varepsilon)\{q^0\}=S_\varepsilon(N).
	\]
	If $M\sim^P N$, preservation of indifference under a common mixture gives $S_\varepsilon(M)\sim^P S_\varepsilon(N)$.  Hence $M\succsim^P N$ implies $S_\varepsilon(M)\succsim^P S_\varepsilon(N)$.  Conversely, if the latter comparison held while $M\not\succsim^P N$, Weak Order would give $N\succ^P M$, and Independence would give $S_\varepsilon(N)\succ^P S_\varepsilon(M)$, a contradiction.  Therefore, for all $M,N\subseteq P$,
	\[
		M\succsim^\ell N\quad\Longleftrightarrow\quad S_\varepsilon(M)\succsim^P S_\varepsilon(N)\quad\Longleftrightarrow\quad M\succsim^P N.
	\]
	This proves the claimed agreement.  Moreover,
	\[
		S_\varepsilon\bigl(\alpha M+(1-\alpha)N\bigr)=\alpha S_\varepsilon(M)+(1-\alpha)S_\varepsilon(N),
	\]
	so Weak Order, vNM Continuity, and Independence are preserved.
		
	By Lemmas S6--S8 in the online supplement to \citet{DLRS2007}, the properties just verified yield an affine Hausdorff-Lipschitz representation $V_P$ of $\succsim^P$ on menus in $P$.  Then $V^\ell(M):=V_P(S_\varepsilon(M))$ is an affine Hausdorff-Lipschitz representation of $\succsim^\ell$, since $d_h(S_\varepsilon(M),S_\varepsilon(N))=\varepsilon d_h(M,N).$
		
	An affine Hausdorff-Lipschitz representation satisfies L-Continuity: if its Lipschitz constant is $L$, choose menus $M^+,M^-$ with $D:=V^\ell(M^+)-V^\ell(M^-)>0$ and take $\gamma\geq L/D$.  The defining inequality of L-Continuity then follows directly from affinity and $d_h(M,N)\leq\alpha/\gamma$.  Thus, $\succsim^\ell$ satisfies the hypotheses of the DLR--DLRS additive theorem.
		
	\medskip
	\noindent
	\emph{\bf Step 3: Apply the DLR--DLRS additive EU representation theorem.}
		
	Let $\mathcal U_0:=\left\{a\in\R^{m+1}\mid \sum_{i=1}^{m+1}a_i=0,\ \sum_{i=1}^{m+1}a_i^2=1\right\}$ be the space of doubly-normalized vNM utility indices over the $m+1$ artificial prizes. The additive representation theorem of \citet{DLR2001}, with the Strong Continuity correction of \citet{DLRS2007}, gives a finite signed Borel measure $\lambda$ on $\mathcal U_0$ such that, up to a positive affine normalization,
	\begin{equation}
	      \overline V(M)=\int_{\mathcal U_0}\left[\max_{q\in M}a\cdot q\right]\,\lambda(da)
	  \label{eq:DLR-rep}
	\end{equation}
	represents $\succsim^\ell$. For $a\in\mathcal U_0$, define $r_i(a):=a_i-a_{m+1}$ for each $i\in\{1,2,\dots,m\}$ and let $r(a):=(r_1(a),\ldots,r_m(a))$.  For $v\in[0,1]^m$,
	\begin{align}
		a\cdot J(v)&=\frac1m\sum_{i=1}^ma_iv_i+a_{m+1}\left(1-\frac1m\sum_{i=1}^mv_i\right) =a_{m+1}+\frac1m r(a)\cdot v.
	\label{eq:artificial-calculation}
	\end{align}
	For a utility menu $C\in\Menus^u$, $J(C):=\{J(v):v\in C\}\subseteq P$. Define $\widetilde V^u(C):=\overline V(J(C)).$ Because $\succsim^\ell$ agrees with $\succsim^P$ on $\Menus^P$ and $\succsim^P$ is the transport of $\succsim^u$ through $J$, we can conclude that $\widetilde V^u$ represents $\succsim^u$ because for all utility menus $C,D\in\Menus^u$,
	\[
		C\succsim^u D\quad\Longleftrightarrow\quad J(C)\succsim^\ell J(D) \quad\Longleftrightarrow\quad \widetilde V^u(C)\geq \widetilde V^u(D).
	\]
	For any utility menu $C\in\Menus^u$, let $h_C:\R^m\to\R$ be defined by $h_C(r):=\max_{v\in C}r\cdot v$.
	\begin{align*}
		\widetilde V^u(C)&=\int_{\mathcal U_0}\left[\max_{v\in C}a\cdot J(v)\right]\,\lambda(da)=\int_{\mathcal U_0}a_{m+1}\,\lambda(da)+\frac1m\int_{\mathcal U_0}h_C(r(a))\,\lambda(da),
	\end{align*}
	Subtracting the menu-independent constant $\int_{\mathcal U_0}a_{m+1}\,\lambda(da)$ produces another representation of $\succsim^u$, namely
	\begin{equation}
		V^u(C)=\frac1m\int_{\mathcal U_0}h_C(r(a))\,\lambda(da).
	\label{eq:direction-rep}
	\end{equation}
	\medskip
	\noindent
	\emph{\bf Step 4: Applying Domination to eliminate negative differences.}
	
	For a utility menu $C\in\Menus^u$, define its \emph{downward closure} as
	\begin{equation}
		C^\downarrow:=\{w\in [0,1]^m\mid \text{there exists }v\in C\text{ such that }w\leq v\}.
	\label{eq:downward-closure}
	\end{equation}
	Thus $C^\downarrow$ adds every utility vector that is state-by-state no better than a vector already in $C$. It is thus compact and a menu.  By \eqref{eq:utility-domination}, adding any individual $w\in C^\downarrow$ to $C$ does not change the menu.  Add a countable dense subset and use Hausdorff Continuity, which follows for $\succsim^u$ from \Cref{lem:affine-lipschitz}, to obtain
	\begin{equation}
		C\sim^u C^\downarrow.
	\label{eq:downward-indifference}
	\end{equation}
		
	For a vector $r=(r_1,\dots,r_m)\in\R^m$, let $r^+:=(r_1^+,\dots, r_m^+)$ where $r_i^+:=\max\{r_i,0\}$. By construction, for any $C\in\Menus^u$ and $r\in\R^m$,
	\begin{equation}
		h_{C^\downarrow}(r)=h_C(r^+).
	\label{eq:positive-part}
	\end{equation}
	Combining \eqref{eq:direction-rep}, \eqref{eq:downward-indifference}, and \eqref{eq:positive-part} gives
	\begin{equation}
		V^u(C)=\frac1m\int_{\mathcal U_0}h_C(r(a)^+)\,\lambda(da).
	\label{eq:positive-directions}
	\end{equation}
		
	\medskip
	\noindent
	\emph{\bf Step 5: Cast subjective states as beliefs.}
		
	For each $a\in\mathcal U_0$, define $s(a):=\sum_{i=1}^m r_i(a)^+.$ Whenever $s(a)>0$, define $P(a):=\frac{r(a)^+}{s(a)}$, then $P(a)\in\DeltaO$ whenever it is defined. 
		
	From $\lambda$, a finite signed Borel measure on $\mathcal U_0$, we can define a signed Borel measure $\eta$ on $\DeltaO$ by
	\begin{equation}
		\eta(E):=\frac1m\int_{\{a\in\mathcal U_0\mid s(a)>0\}}s(a)\mathbf 1\{P(a)\in E\}\,\lambda(da),
	\label{eq:pushforward}
    \end{equation}
	for every Borel set $E\subseteq\DeltaO$. Accordingly, for every bounded Borel function $\varphi$,
	\begin{equation}
		\int_{\DeltaO}\varphi(p)\,\eta(dp)=\frac1m\int_{\{a:s(a)>0\}}s(a)\varphi(P(a))\,\lambda(da).
	\label{eq:pushforward-identity}
	\end{equation}
	Because $\mathcal U_0$ is compact and $s$ is bounded, $\eta$ is a finite signed measure.
		
	Since $h_C(r(a)^+)=s(a)h_C(P(a))$, with both sides equal to zero when $s(a)=0$, equations \eqref{eq:positive-directions} and \eqref{eq:pushforward-identity} imply
	\begin{equation}
		V^u(C)=\int_{\DeltaO}h_C(p)\,\eta(dp).
	\label{eq:posterior-rep}
	\end{equation}
	To see $\eta$ has total mass 1, let $\one$ denote the $m$-dimensional vector whose coordinates all equal to 1, then the restriction of $V^u$ to the singleton menus $\{t\one\}$ is a nonconstant affine representation of the usual order on $[0,1]$. Moreover, $V^u(\{0\one\})=0$ by construction.  Hence there exists $\beta>0$ such that $V^u(\{t\one\})=\beta t$ for any $t\in[0,1]$. Replace $V^u$, $\lambda$, and $\eta$ by $V^u/\beta$,$\lambda/\beta$, and $\eta/\beta$, respectively, and retain the same notation for the rescaled objects.  All the preceding representation equalities remain valid, and we now have
	\begin{equation}
		V^u(\{t\one\})=t\text{ for any }t\in[0,1].
	\label{eq:risk-normalization}
	\end{equation}
	Then: $1=V^u(\{\one\})=\int_{\DeltaO}p\cdot\one\,\eta(dp)=\eta(\DeltaO).$ Finally, for a menu of acts $A\in\Menus$, $h_{T(A)}(p)=\max_{f\in A}\sum_{\omega\in\Omega}p(\omega)u(f(\omega)).$ Using \eqref{eq:pushforward} and \eqref{eq:posterior-rep} yields the normalized additive representation \eqref{eq:NA}.
		
	\medskip
	\noindent
	\emph{``Only if.''} Suppose $(u,\eta)$ is a normalized additive representation.  Weak Order is immediate. Let $\sigma_A^u(p):=\max_{f\in A}\sum_{\omega\in\Omega}p(\omega)u(f(\omega)).$ Then Independence and vNM Continuity follow from the fact that $\sigma_{\alpha A+(1-\alpha)B}^u=\alpha\sigma_A^u+(1-\alpha)\sigma_B^u$ for any $A,B\in\Menus$ and $\alpha\in[0,1]$. Since $u$ is Lipschitz on the finite-dimensional compact lottery simplex, for some $L_u<\infty$,
	\[
	    |W(A)-W(B)|\leq L_u\|\eta\|_{\mathrm{TV}}d_h(A,B).
    \]
	If $W$ is constant, L-Continuity is immediate.  Otherwise, choose $A^+,A^-$ with $D:=W(A^+)-W(A^-)>0$ and take $\gamma\geq L_u\|\eta\|_{\mathrm{TV}}/D$.  If $d_h(A,B)\leq\alpha/\gamma$, then
	\begin{align*}
		&W\bigl((1-\alpha)A+\alpha A^+\bigr)-W\bigl((1-\alpha)B+\alpha A^-\bigr)\\
		&\quad=(1-\alpha)[W(A)-W(B)]+\alpha D\\
		&\quad\geq-(1-\alpha)L_u\|\eta\|_{\mathrm{TV}}d_h(A,B)+\alpha D\geq0.
	\end{align*}
	Thus L-Continuity holds.
		
	For constant lotteries, $W(x)=u(x)\eta(\DeltaO)=u(x),$ so $u$ represents the restriction of $\succsim$ to constant acts.  If $f(\omega)\succsim g(\omega)$ in every state and $f\in A$, then $u(f(\omega))\geq u(g(\omega))$ in every state.  Hence $\sigma_{A\cup\{g\}}^u(p)=\sigma_A^u(p)$ for every $p\in\DeltaO$, which implies $A\sim A\cup\{g\}$.  Domination follows.
	\end{proof}
	
	\section{Uniqueness and Proof}
	
	\begin{proposition}[Uniqueness]\label{prop:normalized-uniqueness}
	Suppose the restriction of $\succsim$ to constant acts is nontrivial. Two normalized additive representations $(u_1,\eta_1)$ and $(u_2,\eta_2)$ both represent $\succsim$ if and only if $u_2=\alpha u_1+\beta$ for some $\alpha>0$ and $\beta\in\R$ and $\eta_1=\eta_2$. In particular, the same uniqueness conclusion holds for additive EU representations.
	\end{proposition}
	
	\begin{proof}
	The ``if'' direction follows immediately from $\eta_1(\DeltaO)=1$. Indeed, if $u_2=\alpha u_1+\beta$ and $\eta_2=\eta_1$, then $W_2=\alpha W_1+\beta$, so the two indices represent the same preference.
		
	We prove the ``only if'' direction. For a constant act $x$, normalization of the measures gives $W_i(x)=u_i(x)$ for $i=1,2$. Thus $u_1$ and $u_2$ are nonconstant affine representations of the same preference over lotteries, so vNM uniqueness gives $u_2=\alpha u_1+\beta$ for some $\alpha>0$ and $\beta\in\R$. Since $\eta_2(\DeltaO)=1$,
	$\frac{W_2(A)-\beta}{\alpha}=\int_{\DeltaO}\left[\max_{f\in A}\sum_{\omega}p(\omega)u_1(f(\omega))\right]\eta_2(dp).$ Thus, replacing $W_2$ by $(W_2-\beta)/\alpha$ leaves $\eta_2$ unchanged and permits us to use the common utility index $u:=u_1$. Normalize it so that $u(X)=[0,1]$.
	
    For $i=1,2$ and $C\in\Menus^u$, define $V_i^u(C):=\int_{\DeltaO}h_C(p)\,\eta_i(dp).$ Because $V_i^u(C)=W_i(\widehat C)$, both $V_1^u$ and $V_2^u$ represent the same relation $\succsim^u$. The mixture-space uniqueness result in Lemma~S7 of the online supplement to \citet{DLRS2007}, applied first to convex utility menus, therefore gives $V_2^u=cV_1^u+d$ for some $c>0$ and $d\in\R$. For every $t\in[0,1]$, $V_i^u(\{t\one\})=\int_{\DeltaO}t\,\eta_i(dp)=t.$ Taking $t=0$ and $t=1$ gives $d=0$ and $c=1$. Since every menu has the same support function as its convex hull, it follows that
	\begin{equation}
		V_1^u(C)=V_2^u(C) \text{ for every }C\in\Menus^u.
	\label{eq:equal-utility-indices}
	\end{equation}
		
	We now use the objects introduced in Steps~2--5 in the proof of Proposition S1 to translate each $\eta_i$ into a DLRS measure on $\mathcal U_0$. For $p=(p_1,\ldots,p_m)\in\DeltaO$, let
	\[
	    z(p):=\left(p_1-\frac1{m+1},\ldots,p_m-\frac1{m+1},-\frac1{m+1}\right),
		\qquad
		\kappa(p):=\|z(p)\|_2,
	\]
	and define
	$\Phi(p):=\frac{z(p)}{\kappa(p)}.$ The coordinates of $z(p)$ sum to zero and its final coordinate is nonzero; hence $\kappa(p)>0$ and $\Phi(p)\in\mathcal U_0$. Moreover, using the definitions of $r$, $s$, and $P$ from Steps~3 and~5,
	\begin{equation}
		r(\Phi(p))=\frac{p}{\kappa(p)},
		\qquad
		s(\Phi(p))=\frac1{\kappa(p)},
		\qquad
		P(\Phi(p))=p.
	\label{eq:Phi-properties}
	\end{equation}
		
	Define a finite signed Borel measure $\lambda_i$ on $\mathcal U_0$ by
	\begin{equation}
		\lambda_i(E):=m\int_{\DeltaO}\kappa(p)\mathbf 1\{\Phi(p)\in E\}\,\eta_i(dp),
	\label{eq:uniqueness-lambda}
	\end{equation}
	for every Borel set $E\subseteq\mathcal U_0$. The map $\Phi$ is continuous, so $\lambda_i$ is a Borel measure; finiteness follows because $\kappa$ is bounded on $\DeltaO$. By \eqref{eq:Phi-properties},
	\begin{align}
		\frac1m\int_{\mathcal U_0}h_C(r(a))\,\lambda_i(da)&=\int_{\DeltaO}\kappa(p)h_C\left(\frac{p}{\kappa(p)}\right)\eta_i(dp)=\int_{\DeltaO}h_C(p)\,\eta_i(dp)=V_i^u(C).
	\label{eq:eta-lambda-identity}
	\end{align}
		
	For a lottery menu $M\in\Menus^\ell$, set $\overline V_i(M):=\int_{\mathcal U_0}\left[\max_{q\in M}a\cdot q\right]\lambda_i(da).$ The identity $a\cdot J(v)=a_{m+1}+m^{-1}r(a)\cdot v$ established in Step~3, together with \eqref{eq:eta-lambda-identity}, implies that, for every utility menu $C$,
	\begin{equation}
		\overline V_i(J(C))=\int_{\mathcal U_0}a_{m+1}\,\lambda_i(da)+V_i^u(C).
	\label{eq:lifted-index-on-P}
	\end{equation}
	The first term is independent of $C$, so $\overline V_i$ represents $\succsim^P$ on menus in $P$. Furthermore,
	\[
	    \overline V_i(S_\varepsilon(M))=\varepsilon\overline V_i(M)+(1-\varepsilon)\overline V_i(\{q^0\}).
	\]
	We also have that $\overline V_i$ represents $\succsim^\ell$ on $\Menus^\ell$ since, for all $M,N\in\Menus^\ell$,
	\begin{align*}
		M\succsim^\ell N&\iff S_\varepsilon(M)\succsim^P S_\varepsilon(N)\iff \overline V_i(S_\varepsilon(M))\geq\overline V_i(S_\varepsilon(N))\iff \overline V_i(M)\geq\overline V_i(N).
	\end{align*}
		
	Both $\overline V_1$ and $\overline V_2$ are therefore affine representations of the same nontrivial preference. Applying Lemma~S7 in the DLRS supplement to convex lottery menus gives $\overline V_2=c\overline V_1+d$ for some $c>0$ and $d\in\R$ on that domain. Since each $\overline V_i$ assigns the same value to a menu and its convex hull, the equality extends to all lottery menus. Since every $a\in\mathcal U_0$ satisfies $a\cdot q^0=0$, we have $\overline V_1(\{q^0\})=\overline V_2(\{q^0\})=0$, and hence $d=0$. In addition, \eqref{eq:lifted-index-on-P} and \eqref{eq:equal-utility-indices} give $\overline V_i(J(\{\one\}))-\overline V_i(J(\{0\one\}))=V_i^u(\{\one\})-V_i^u(\{0\one\})=1.$ Therefore $c=1$ and $\overline V_1=\overline V_2$.
		
	We can now apply the canonical-measure uniqueness argument in the DLRS supplement. Let $\nu:=\lambda_1-\lambda_2$. The preceding equality implies that integration against $\nu$ vanishes on every menu support function on $\mathcal U_0$. It therefore vanishes on the linear space generated by differences of such support functions. Lemma~S10(3) of the supplement shows that this space is dense in $C(\mathcal U_0)$. Because integration against a finite signed measure is continuous in the supremum norm, it follows that integration against $\nu$ vanishes on all of $C(\mathcal U_0)$. Equivalently, this is the unique-extension argument of Lemma~S11. Uniqueness in the Riesz representation used in Lemma~S12 then gives $\lambda_1=\lambda_2$.
		
	Finally, the weighted pushforward already introduced in Step~5 recovers $\eta_i$ from $\lambda_i$. Indeed, by \eqref{eq:uniqueness-lambda} and \eqref{eq:Phi-properties}, for every Borel set $E\subseteq\DeltaO$,
	\begin{align*}
		\frac1m\int_{\{a:s(a)>0\}}s(a)\mathbf 1\{P(a)\in E\}\,\lambda_i(da)&=\int_{\DeltaO}\kappa(p)\frac1{\kappa(p)}\mathbf 1\{p\in E\}\,\eta_i(dp)=\eta_i(E).
	\end{align*}
	Thus the map $\lambda_i\mapsto\eta_i$ is exactly the map in \eqref{eq:pushforward}. The fact that $\lambda_1=\lambda_2$ consequently implies $\eta_1=\eta_2$, completing the proof.
	\end{proof}


\begin{thebibliography}{99}
\emergencystretch=1em

\bibitem[Ahn(2008)]{Ahn2008}
Ahn, David S. 2008.
``Ambiguity Without a State Space.''
\emph{Review of Economic Studies} 75(1): 3--28.

\bibitem[Ahn and Sarver(2013)]{AhnSarver2013}
Ahn, David S., and Todd Sarver. 2013.
``Preference for Flexibility and Random Choice.''
\emph{Econometrica} 81(1): 341--361.

\bibitem[Aliprantis and Border(2006)]{AliprantisBorder2006}
Aliprantis, Charalambos D., and Kim C. Border. 2006.
\emph{Infinite Dimensional Analysis: A Hitchhiker's Guide}.
3rd ed. Berlin: Springer.

\bibitem[Bell(1982)]{Bell1982}
Bell, David E. 1982.
``Regret in Decision Making under Uncertainty.''
\emph{Operations Research} 30(5): 961--981.

\bibitem[Blumstein, McMichael, and Storrow(2020)]{BlumsteinEtAl2020}
Blumstein, James F., Benjamin J. McMichael, and Alan B. Storrow. 2020.
``Constraints on Medical Liability Through Malpractice Safe Harbors.''
\emph{JAMA Health Forum} 1(8): e200961.

\bibitem[Blumstein, McMichael, and Storrow(2023)]{BlumsteinEtAl2023}
Blumstein, James F., Benjamin J. McMichael, and Alan B. Storrow. 2023.
``Developing Safe Harbors to Address Malpractice Liability and Wasteful Health
Care Spending.''
\emph{JAMA Health Forum} 4(11): e233899.

\bibitem[Cram\'er and Wold(1936)]{CramerWold1936}
Cram\'er, Harald, and Herman Wold. 1936.
``Some Theorems on Distribution Functions.''
\emph{Journal of the London Mathematical Society} s1-11(4): 290--294.

\bibitem[Dekel and Lipman(2012)]{DekelLipman2012}
Dekel, Eddie, and Barton L. Lipman. 2012.
``Costly Self-Control and Random Self-Indulgence.''
\emph{Econometrica} 80(3): 1271--1302.

\bibitem[Dekel, Lipman, and Rustichini(2001)]{DekelLipmanRustichini2001}
Dekel, Eddie, Barton L. Lipman, and Aldo Rustichini. 2001.
``Representing Preferences with a Unique Subjective State Space.''
\emph{Econometrica} 69(4): 891--934.

\bibitem[Dekel et al.(2007)]{DekelLipmanRustichiniSarver2007}
Dekel, Eddie, Barton L. Lipman, Aldo Rustichini, and Todd Sarver. 2007.
``Representing Preferences with a Unique Subjective State Space: A
Corrigendum.''
\emph{Econometrica} 75(2): 591--600.

\bibitem[de Oliveira et al.(2017)]{deOliveiraDentiMihmOzbek2017}
de Oliveira, Henrique, Tommaso Denti, Maximilian Mihm, and Kemal \"Ozbek. 2017.
``Rationally Inattentive Preferences and Hidden Information Costs.''
\emph{Theoretical Economics} 12(2): 621--654.

\bibitem[Dillenberger et al.(2014)]{dillenberger2014theory}
Dillenberger, David, Juan Sebasti\'{a}n Lleras, Philipp Sadowski, and Norio
Takeoka. 2014.
``A Theory of Subjective Learning.''
\emph{Journal of Economic Theory} 153: 287--312.

\bibitem[Ederer and Manso(2013)]{EdererManso2013}
Ederer, Florian, and Gustavo Manso.\penalty-100\ 2013.
``Is Pay for Per\-formance Detrimental to Innovation?''
\emph{Management Science} 59(7): 1496--1513.

\bibitem[Fischhoff(1975)]{Fischhoff1975}
Fischhoff, Baruch. 1975.
``Hindsight $\neq$ Foresight: The Effect of Outcome Knowledge on Judgment under Uncertainty.''
\emph{Journal of Experimental Psychology: Human Perception and Performance}
1(3): 288--299.

\bibitem[Folland(1999)]{Folland1999}
Folland, Gerald B. 1999.
\emph{Real Analysis: Modern Techniques and Their Applications}.
2nd ed. New York: Wiley.

\bibitem[Frakes and Gruber(2019)]{FrakesGruber2019}
Frakes, Michael, and Jonathan Gruber. 2019.
``Defensive Medicine: Evidence from Military Immunity.''
\emph{American Economic Journal: Economic Policy} 11(3): 197--231.

\bibitem[Gul and Pesendorfer(2001)]{GulPesendorfer2001}
Gul, Faruk, and Wolfgang Pesendorfer. 2001.
``Temptation and Self-Control.''
\emph{Econometrica} 69(6): 1403--1435.

\bibitem[Gul and Pesendorfer(2006)]{GulPesendorfer2006}
Gul, Faruk, and Wolfgang Pesendorfer. 2006.
``Random Expected Utility.''
\emph{Econometrica} 74(1): 121--146.

\bibitem[Hawkins and Hastie(1990)]{HawkinsHastie1990}
Hawkins, Scott A., and Reid Hastie. 1990.
``Hindsight: Biased Judgments of Past Events after the Outcomes Are Known.''
\emph{Psychological Bulletin} 107(3): 311--327.

\bibitem[Higashi et al.(2009)]{HigashiHyogoTakeoka2009}
Higashi, Youichiro, Kazuya Hyogo, and Norio Takeoka. 2009.
``Subjective Random Discounting and Intertemporal Choice.''
\emph{Journal of Economic Theory} 144(3): 1015--1053.

\bibitem[Kamin and Rachlinski(1995)]{KaminRachlinski1995}
Kamin, Kim A., and Jeffrey J. Rachlinski. 1995.
``Ex Post $\neq$ Ex Ante: Determining Liability in Hindsight.''
\emph{Law and Human Behavior} 19(1): 89--104.

\bibitem[Kreps(1979)]{Kreps1979}
Kreps, David M. 1979.
``A Representation Theorem for `Preference for Flexibility'.''
\emph{Econometrica} 47(3): 565--577.

\bibitem[Kopylov(2012)]{Kopylov2012}
Kopylov, Igor. 2012.
``Perfectionism and Choice.''
\emph{Econometrica} 80(5): 1819--1843.

\bibitem[Kripfganz and Schulze(1987)]{KripfganzSchulze1987}
Kripfganz, A., and R. Schulze. 1987.
``Piecewise Affine Functions as a Difference of Two Convex Functions.''
\emph{Optimization} 18(1): 23--29.

\bibitem[Loomes and Sugden(1982)]{LoomesSugden1982}
Loomes, Graham, and Robert Sugden. 1982.
``Regret Theory: An Alternative Theory of Rational Choice under
Uncertainty.''
\emph{Economic Journal} 92(368): 805--824.

\bibitem[Lu(2016)]{Lu2016}
Lu, Jay. 2016.
``Random Choice and Private Information.''
\emph{Econometrica} 84(6): 1983--2027.

\bibitem[Manso(2011)]{Manso2011}
Manso, Gustavo. 2011.
``Motivating Innovation.''
\emph{Journal of Finance} 66(5): 1823--1860.

\bibitem[Pennesi(2026)]{Pennesi2026}
Pennesi, Daniele. 2026.
``Subjective Timing of Information Arrival.''
Working paper, University of Torino.

\bibitem[Rockafellar and Wets(1998)]{RockafellarWets1998}
Rockafellar, R. Tyrrell, and Roger J.-B. Wets. 1998.
\emph{Variational Analysis}. Berlin: Springer.

\bibitem[Sarver(2008)]{Sarver2008}
Sarver, Todd. 2008.
``Anticipating Regret: Why Fewer Options May Be Better.''
\emph{Econometrica} 76(2): 263--305.

\bibitem[Tian and Wang(2014)]{TianWang2014}
Tian, Xuan, and Tracy Yue Wang. 2014.
``Tolerance for Failure and Corporate Innovation.''
\emph{Review of Financial Studies} 27(1): 211--255.

\end{thebibliography}

\begin{thebibliography}{9}
		
		\bibitem[Dekel, Lipman, and Rustichini(2001)]{DLR2001}
		Dekel, E., B. L. Lipman, and A. Rustichini (2001),
		``Representing Preferences with a Unique Subjective State Space,''
		\emph{Econometrica} 69, 891--934.
		
		\bibitem[Dekel, Lipman, Rustichini, and Sarver(2007)]{DLRS2007}
		Dekel, E., B. L. Lipman, A. Rustichini, and T. Sarver (2007),
		``Representing Preferences with a Unique Subjective State Space: A Corrigendum,''
		\emph{Econometrica} 75, 591--600, together with its online supplement.
		
		\bibitem[Dillenberger, Lleras, Sadowski, and Takeoka(2014)]{DLST2014}
		Dillenberger, D., J. S. Lleras, P. Sadowski, and N. Takeoka (2014),
		``A Theory of Subjective Learning,''
		\emph{Journal of Economic Theory} 153, 287--312.
		
	\end{thebibliography}
\end{document}